\documentclass[sigconf]{acmart}
\AtBeginDocument{%
  }

\copyrightyear{2026}
\acmYear{2026}
\setcopyright{cc}
\setcctype{by}
\acmConference[CIKM '26]
  {Proceedings of the 35th ACM International Conference on Information and Knowledge Management}
  {November 07--11, 2026}
  {Rome, Italy}
\acmBooktitle{Proceedings of the 35th ACM International Conference on Information and Knowledge Management (CIKM '26), November 07--11, 2026, Rome, Italy}
\acmDOI{10.1145/3799682.3840119}
\acmISBN{979-8-4007-2539-5/2026/11}
\usepackage{multirow}
\usepackage{makecell}
\usepackage{enumitem}
\usepackage{subfigure}
\theoremstyle{definition}
\newtheorem{definition}{Definition}[section]
\newtheorem{assumption}{Assumption}[section]
\newtheorem{theorem}{Theorem}[section]
\begin{document}

% \title{A Unified Causal Framework for Multi-valued Treatment: Jointly Optimizing Debiased CTR and Uplift}
\title{Jointly Optimizing Debiased CTR and Uplift for Coupons Marketing: A Unified Causal Framework}

\author{Siyun Yang}
\affiliation{
  \institution{Kuaishou Technology}
  \city{Beijing}
  \country{China}
}
\email{yangsiyun@kuaishou.com}

\author{Shixiao Yang}
\affiliation{
  \institution{Beijing Institute of Technology}
  \city{Beijing}
  \country{China}
}
\email{ysx144\_51@bit.edu.cn}

\author{Jian Wang}
\authornote{Project leader.}
\affiliation{
  \institution{Kuaishou Technology}
  \city{Beijing}
  \country{China}
}
\email{wangjian27@kuaishou.com}

\author{Di Fan}
\affiliation{
  \institution{Kuaishou Technology}
  \city{Beijing}
  \country{China}
}
\email{fandi@kuaishou.com}

\author{Kehe Cai}
\affiliation{
  \institution{Kuaishou Technology}
  \city{Beijing}
  \country{China}
}
\email{caikehe@kuaishou.com}

\author{Haoyan Fu}
\affiliation{
  \institution{Beijing Institute of Technology}
  \city{Beijing}
  \country{China}
}
\email{haoyan-fu@bit.edu.cn}

\author{Jiaming Zhang}
\affiliation{
  \institution{Kuaishou Technology}
  \city{Beijing}
  \country{China}
}
\email{zhangjiaming07@kuaishou.com}

\author{Wenjin Wu}
\authornote{Corresponding author.}
\affiliation{
  \institution{Kuaishou Technology}
  \city{Beijing}
  \country{China}
}
\email{wuwenjin@kuaishou.com}

\author{Peng Jiang}
\affiliation{
  \institution{Independent Researcher}
  \city{Beijing}
  \country{China}
}
\email{13126980773@139.com}

\renewcommand{\shortauthors}{Siyun Yang et al.}

\begin{abstract}
% KDD 
In online advertising, marketing interventions such as coupons introduce significant confounding bias into Click-Through Rate (CTR) prediction. Observed clicks reflect a mixture of users' intrinsic preferences and the uplift induced by these interventions. This causes conventional models to miscalibrate base CTRs, which distorts downstream ranking and billing decisions. Furthermore, marketing interventions often operate as multi-valued treatments with varying magnitudes, introducing additional complexity to CTR prediction.

To address these issues, we propose the \textbf{Uni}fied \textbf{M}ulti-\textbf{V}alued \textbf{T}reatment Network (UniMVT). Specifically, UniMVT disentangles confounding factors from treatment-sensitive representations, enabling a full-space counterfactual inference module to jointly reconstruct the debiased base CTR and intensity-response curves. To handle the complexity of multi-valued treatments, UniMVT employs an auxiliary intensity estimation task to capture treatment propensities and devise a unit uplift objective that normalizes the intervention effect. This ensures comparable estimation across the continuous coupon-value spectrum. UniMVT simultaneously achieves debiased CTR prediction for accurate system calibration and precise uplift estimation for incentive allocation. Extensive experiments on synthetic and industrial datasets demonstrate UniMVT's superiority in both predictive accuracy and calibration. Furthermore, real-world A/B tests confirm that UniMVT significantly improves business metrics through more effective coupon distribution.

\end{abstract}






\maketitle
% \fancyhead{} % 清空所有页眉设置

\section{Introduction}

Click-Through Rate (CTR) prediction~\cite{ma2018entire, yang2021multi, mao2023finalmlp, guo2017deepfm, wang2022enhancing, chang2023twin, fu2025unified}  remains a cornerstone of modern recommendation and advertising systems. Its evolution has been significantly driven by Multi-Task Learning (MTL)~\cite{caruana1997multitask, li2023removing, wang2023single, zhao2019recommending, bi2022mtrec, li2023adatt, su2024stem, jiang2024automatic} architectures, such as MMoE~\cite{ma2018modeling} and PLE~\cite{tang2020progressive}, which exploit correlations across diverse user behaviors to improve overall performance. However, modern industrial systems increasingly deploy proactive marketing interventions, such as discount coupons and personalized incentives, to further stimulate user engagement. As illustrated in Figure~\ref{fig:coupon_to_ctr}, the presence of coupons significantly shifts the observed click probability for identical user-item pairs compared to the untreated scenario. In these cases, the observed click signal is a mixture of users’ intrinsic preferences and the uplift induced by these interventions. This entanglement is further complicated by the inherent heterogeneity in user responses: high-propensity users may convert based on interest alone, making incentives redundant, while price-sensitive users often require substantial discounts. Although stronger incentives generally boost conversion, they also incur higher costs. Accurately capturing these heterogeneous treatment effects is therefore essential for moving beyond traditional CTR prediction toward a more nuanced, incentive-aware recommendation paradigm.
% Click-through rate (CTR) prediction~\cite{ma2018entire, yang2021multi, mao2023finalmlp, guo2017deepfm, wang2022enhancing, chang2023twin, fu2025unified} has evolved significantly through multi-task learning architectures~\cite{caruana1997multitask, li2023removing, wang2023single, zhao2019recommending, bi2022mtrec, li2023adatt, su2024stem, jiang2024automatic}, such as MMoE~\cite{ma2018modeling} and PLE~\cite{tang2020progressive}, which exploit correlations across diverse user behaviors. However, modern industrial systems increasingly deploy proactive marketing interventions, such as discount coupons and personalized incentives, to further stimulate user engagement. As illustrated in Figure~\ref{fig:coupon_to_ctr}, the presence of coupons significantly shifts the observed click probability for identical user-item pairs compared to the untreated scenario. In these cases, the observed click signal is no longer a pure reflection of a user's intrinsic preference, user responses to such interventions are inherently heterogeneous. High-propensity users may engage with a specific advertisement based on baseline interest alone, regardless of the incentive level. Conversely, price-sensitive individuals typically require substantial discounts to be motivated. While increasing discount intensities generally correlates with higher conversion probabilities, it simultaneously incurs escalating operational costs, necessitating a strategic balance between engagement uplift and budget expenditure.  

\begin{figure*}[ht]
  \centering
  \subfigure[The industrial scenario: Comparison of CTR with and without coupon interventions. \label{fig:coupon_to_ctr}]{
    \includegraphics[width=0.24\linewidth]{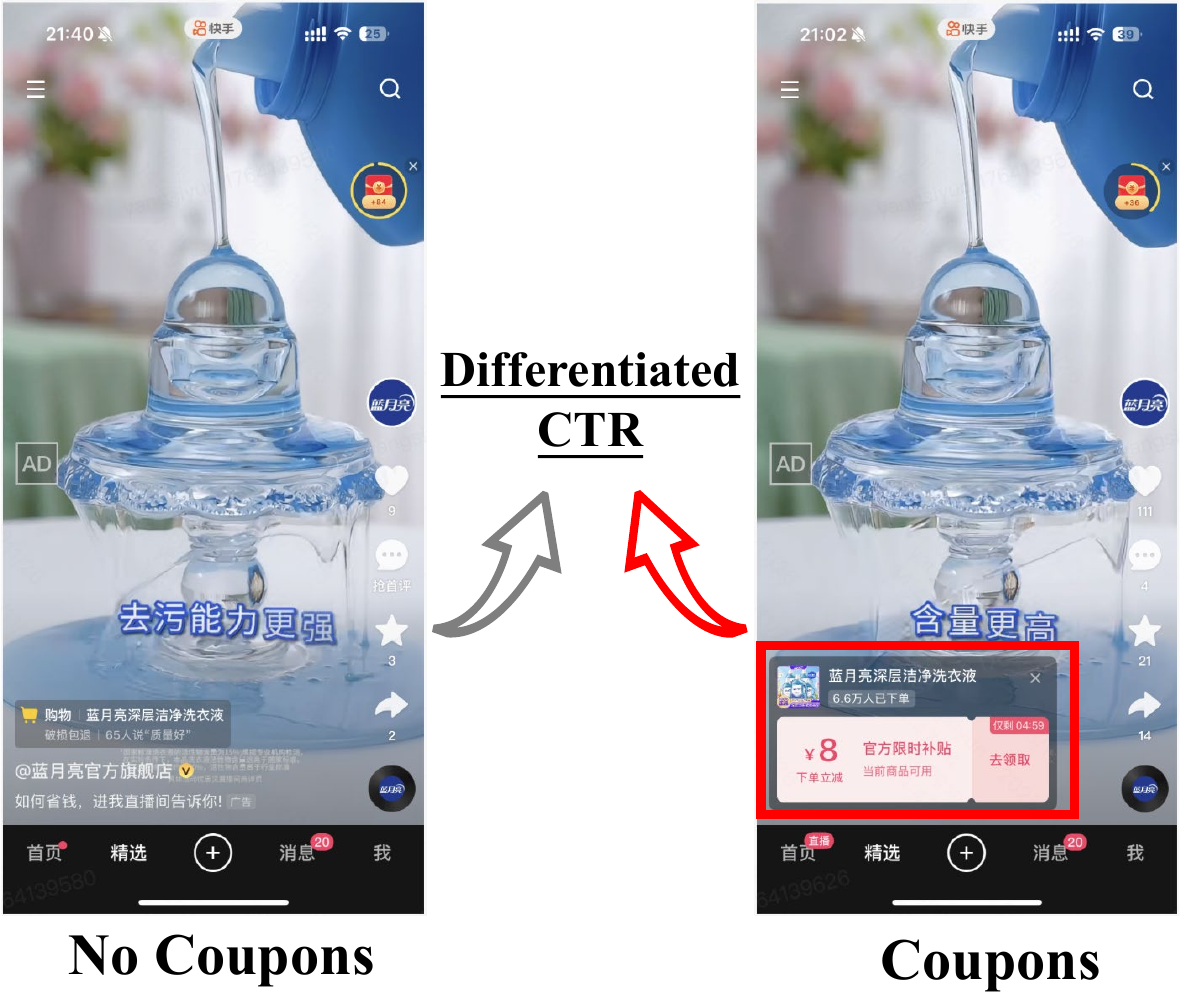}}
    \hspace{0.01\linewidth}
  % \subfigure[Heterogeneous user sensitivities: Heterogeneous user reactions to various discount intensities and associated cost tradeoffs.\label{fig:heterogeneous_example}]{
  %   \includegraphics[width=0.37\linewidth]{pics/hamburger.pdf}}
  %   \hspace{0.01\linewidth}
  % \subfigure[The empirical analysis: Baseline/ treatment-agnostic CTR models PCOC and sample density distribution across treatment intensities on Randomized Controlled Trial (RCT) samples.\label{fig:pcoc_dose_plot}]{
  %   \centering
  %   \includegraphics[width=0.66\linewidth]{pics/new_flow.pdf}}
  % \caption{Motivation for multi-value causal marketing interventions. }\label{fig:motivation_combined}
  \subfigure[Architectural limitations of prevailing uplift frameworks. (i) Baseline1: Two-Stage separate model. (ii) Baseline2: Unified MTL with Shared Bottom model.\label{fig:baselines}]{
    \centering
    \includegraphics[width=0.66\linewidth]{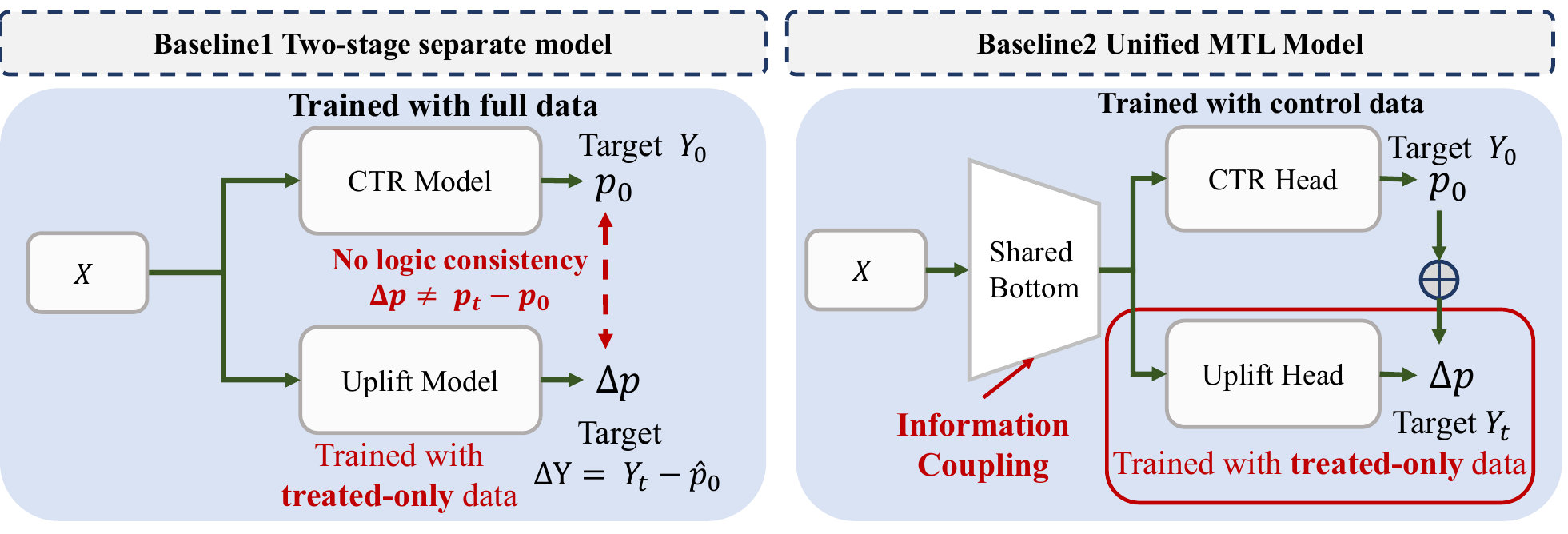}}
    \vspace{-2em}
  \caption{Motivation for multi-value causal marketing interventions. }\label{fig:motivation_combined}
  \Description{Comparison of coupon and no-coupon scenarios and of two baseline uplift-model architectures.}
  \vspace{-1em}
\end{figure*}

% Uplift modeling provides the theoretical foundation for estimating heterogeneous treatment effects (HTE). Foundational meta-learning frameworks, such as S-Learner and T-Learner~\cite{kunzel2019metalearners}, utilize standard supervised models to isolate variations across treatment groups. To mitigate selection bias in observational data, CFRNet~\cite{shalit2017estimating} minimizes distribution discrepancies in a latent space, while FlexTENet~\cite{curth2021inductive} shares information across potential outcomes via inductive biases. Recently, DESCN~\cite{zhong2022descn} employs cross-networks to alleviate data sparsity across the entire sample space. 
% Moreover, in industrial systems, uplift modeling has become vital for optimizing complex decision-making. To handle multi-valued and continuous treatments, recent applications explicitly integrate intervention attributes—such as EFIN~\cite{liu2023explicitfeatureinteractionawareuplift} for fine-grained coupon distribution, or continuous causal forests for dynamic pricing~\cite{wan2022gcfgeneralizedcausalforest}. Furthermore, state-of-the-art methods have evolved beyond binary outcomes to accommodate complex, real-world objectives. Recent advancements tackle long-tail continuous revenue uplift via Zero-Inflated Lognormal (ZILN)~\cite{he2024rankabilityenhancedrevenueupliftmodeling} distributions and rank-aware losses, while dynamic frameworks like CDUM~\cite{meng2024coarsetofinedynamicupliftmodeling} capture real-time contexts for personalized interventions in streaming feeds.

Uplift modeling provides the theoretical foundation for estimating heterogeneous treatment effects. Specifically, S-Learner and T-Learner~\cite{kunzel2019metalearners} serve as foundational meta-learning frameworks that utilize standard supervised models to isolate variations across different treatment groups. To mitigate selection bias in observational data, CFRNet~\cite{shalit2017estimating} introduces representation balancing to minimize the distribution discrepancy between treated and control populations in a latent space, while FlexTENet~\cite{curth2021inductive} adaptively shares information across potential outcomes to exploit structural similarities through flexible inductive biases. To handle continuous or multi-valued treatment intensities, DRNet~\cite{schwab2020learning} partitions the intensity space into discrete intervals for localized estimation, whereas VCNet~\cite{nie2021vcnet} leverages varying coefficient networks to ensure the functional continuity of estimated intensity-response curves. Additionally, CEVAE approximates latent factors that confound treatment assignments through variational autoencoders, and more recently, DESCN~\cite{zhong2022descn} employs cross-networks to perform estimation across the entire sample space to alleviate data sparsity.

% 接下来可以接上您的转折（The Gap）：
% Despite these advancements, directly applying existing frameworks to continuous industrial interventions reveals two critical gaps:
% \textbf{1) Strategy Side:} The Revenue-Cost Trade-off Dilemma. Optimizing coupon delivery requires balancing incremental revenue against continuous intervention costs. However, existing methods typically resort to piecewise modeling (e.g., discretization) for varying intensities, lacking a consistent, global structural assumption for continuous treatments. Without a unified formulation that inherently respects diminishing marginal returns, downstream optimization fails to evaluate intermediate pricing tiers effectively, frequently collapsing into suboptimal allocations.
% \textbf{2) Modeling Side:} Confounding Bias and Calibration Drift. Due to selection bias in historical logs, treatment assignments are heavily confounded by user features. As shown in Figure \ref{fig:baselines}, Conventional approaches often deploy uplift modeling as an isolated downstream module or train exclusively on treated samples. Failing to disentangle intrinsic user preferences from treatment-sensitive responses across the entire sample space induces confounding bias and calibration drift, yielding unreliable base CTR and uplift predictions.
Despite these advancements, applying existing frameworks to continuous industrial interventions reveals two critical limitations:
\textbf{1) Architectural isolation and data fragmentation:} As shown in Figure \ref{fig:baselines}(i), conventional approaches often deploy uplift modeling as an isolated downstream module trained exclusively on treated samples. Consequently, the upstream CTR model remains blind to interventions, yielding biased predictions that undermine downstream calibration. Furthermore, restricting uplift training to this narrow sub-population exacerbates data sparsity and hinders the learning of robust user representations.
\textbf{2) Confounding bias and calibration drift:} As illustrated in Figure \ref{fig:baselines}(ii), even when CTR and uplift predictions are integrated into a unified model, historical treatment assignments remain heavily confounded by user features due to selection bias. Failing to explicitly disentangle intrinsic user preferences from treatment-sensitive responses across the entire sample space induces calibration drift, yielding unreliable baseline CTR and uplift predictions.

To overcome these dual bottlenecks, it is imperative to integrate CTR debiasing and continuous uplift estimation into a unified, entire-space framework.
To address these challenges, we propose a unified causal modeling paradigm named UniMVT (\textbf{Uni}fied \textbf{M}ulti-\textbf{V}alued \textbf{T}reatment Network). At the architectural level, UniMVT employs a Deconfounded Causal Representation (DCR) layer to disentangle invariant confounding variables from treatment-sensitive features. Methodologically, we decouple the estimation of base CTR and incremental uplift by introducing a Counterfactual X-Network. Crucially, we propose a Latent Unit Sensitivity (uCATE) formulation. By enforcing a latent monotonic linear treatment shift strictly within the logit space, UniMVT intrinsically captures the non-linear saturation effect of continuous intervention, providing robust and granular criteria for optimal ROI-constrained coupon allocation.

Our main contributions are summarized as follows:
\begin{itemize}[leftmargin=8pt]
\item \textbf{Unified Inference Paradigm}: We develop a treatment-aware CTR inference framework incorporating monotonic constraints to capture the heterogeneity of user sensitivities across varying treatment intensities, enabling simultaneous global CTR debiasing and precise uplift estimation. 
\item \textbf{Causal Framework}: We develop a causal full-space framework featuring a Deconfounded Causal Representation (DCR) layer to isolate confounding factors, and a Counterfactual X-Network that mutually regularizes factual and counterfactual predictions to bridge the observational gap.Theoretical analysis demonstrates the framework's convergence.
\item \textbf{Industrial Deployment}: Extensive experiments on both robust synthetic benchmarks and a massive real-world dataset demonstrate UniMVT's state-of-the-art performance. 
\end{itemize}

\section{Related Works}

Uplift models for CATE estimation have evolved from meta learning to deep architectures. Early representation-based methods, such as CFRNet~\cite{shalit2017estimating} and FlexTENet~\cite{curth2021inductive}, focused on minimizing distribution discrepancies, while X-learner~\cite{kunzel2019metalearners} addressed treatment imbalance. To mitigate biases in large-scale observational data, approaches like DESCN~\cite{zhong2022descn} and EUEN~\cite{ke2021addressing} employ entire-space modeling to jointly learn propensity and response functions. For continuous treatments, research has progressed from partition-based methods like DRNet~\cite{schwab2020learning} to functional continuity models like VCNet~\cite{nie2021vcnet}. Furthermore, recent studies have extended these frameworks to handle complex industrial constraints, including temporal dynamics~\cite{zhang2024temporal}, revenue rankability~\cite{he2024rankability}, and robust estimation under covariate shifts~\cite{wang2025inter,tao2023event,ai2024improve}.
In industrial systems, recent applications integrate intervention attributes, as seen in EFIN~\cite{liu2023explicitfeatureinteractionawareuplift} for coupon distribution and continuous causal forests for dynamic pricing~\cite{wan2022gcfgeneralizedcausalforest}. Furthermore, methods have evolved beyond binary outcomes to handle complex objectives: addressing long-tail continuous revenue uplift via ZILN distributions and rank-aware losses~\cite{he2024rankabilityenhancedrevenueupliftmodeling}, and capturing real-time contexts for personalized streaming interventions via frameworks like CDUM~\cite{meng2024coarsetofinedynamicupliftmodeling}.

\section{Preliminaries}

\subsection{Notations and Definitions}

We formally define the observed dataset as $\mathcal{D} = \{(\mathbf{x}_i, w_i, d_i, y_i)\}_{i=1}^N$, where:

\begin{itemize}[leftmargin=12pt]
\item $\mathbf{x}_i \in \mathcal{X} \subseteq \mathbb{R}^k$ denotes the context feature vector, encompassing user attributes and historical behaviors.
\item $w_i \in \{0, 1\}$ is the binary treatment indicator, where $w_i=1$ signifies exposure to the intervention and $w_i=0$ indicates the control status.
\item $t_i \in \mathcal{T} \subset \mathbb{R}^+$ represents the treatment intensity (e.g., coupon face value or discount magnitude). Note that $t_i$ is observable and non-zero if and only if $w_i=1$; otherwise, we define $t_i=0$.
\item $y_i \in \{0, 1\}$ is the binary click label.
\end{itemize}

\subsection{Causal Assumptions}

To accurately quantify the impact of interventions on Click-Through Rate (CTR) variations, we frame the problem within the causal inference paradigm. We establish the identifiability of the Conditional Average Treatment Effect (CATE) through standard assumptions and introduce a structural constraint to model the intensity-response relationship effectively.

\begin{assumption}[Consistency, unconfoundedness and overlap]
Following the standard potential outcomes framework~\cite{Rubin01032005}, we adopt the following assumptions to ensure the identifiability of the treatment effect:

\begin{itemize}[leftmargin=8pt]
\item \textbf{Consistency}: The observed outcome $y$ for a unit receiving treatment $t$ is identical to the potential outcome $Y(t)$; formally, $Y = Y(t)$ if $T=t$.
\item \textbf{Unconfoundedness}: Given the covariate vector $\mathbf{x}$, the treatment assignment $T$ is independent of the potential outcomes; formally, $\{Y(t)\}_{t \in \mathcal{T}} \perp\!\!\!\perp T \mid X$.
\item \textbf{Overlap}: Every unit has a non-zero probability of receiving any treatment level within the domain; formally, $0 < P(T=t \mid X=\mathbf{x}) < 1$ for all $t \in \mathcal{T}$ and $\mathbf{x} \in \mathcal{X}$.
\end{itemize}
\end{assumption}

% \begin{assumption}[Monotonic Linear Relationship]\label{assump:monotonicity}
% Motivated by empirical statistics from randomized experiments\ref{fig:motivation_combined}, we posit that the incremental uplift exhibits a monotonic linear relationship with the treatment intensity $t$. Formally, for a user with features $\mathbf{x}$, the CATE is proportional to the intensity:
% \begin{equation}
% \tau(\mathbf{x}, t) \approx \alpha(\mathbf{x}) \cdot t, \quad \text{where } \alpha(\mathbf{x}) > 0.
% \end{equation}
% This assumption implies that the marginal gain is constant per unit of investment. This formulation not only aligns with the observed intensity-response patterns in industrial settings but also significantly simplifies the optimization landscape, allowing for the direct calculation of ROI and the efficient derivation of optimal coupon face values.
% \end{assumption}

\begin{assumption}[Latent Linear Monotonicity]\label{assump:latent_linearity}
Motivated by empirical statistics from randomized experiments (Figure \ref{fig:motivation_combined}) and the diminishing marginal utility commonly observed in economics, we avoid imposing a strict linear assumption on the absolute probability scale, which is inherently bounded by $[0,1]$. Instead, we postulate a structural assumption in the latent space: the treatment intensity $t$ exerts a monotonic linear shift on the user's conversion intention within the log-odds (logit) space. Formally, for a user with features $\mathbf{x}$, there exists a user-specific positive scalar $\eta(\mathbf{x}) > 0$ such that the post-treatment logit is shifted proportionally by $t$:
\begin{equation}
\text{logit}(p_t(\mathbf{x})) - \text{logit}(p_0(\mathbf{x})) = \eta(\mathbf{x}) \cdot t, \quad \text{where } \eta(\mathbf{x}) > 0.
\end{equation}
% Consequently, the Conditional Average Treatment Effect (CATE) in the probability space inherently models the non-linear saturation effect:
% \begin{equation}
% \tau(\mathbf{x}, t) = \sigma(\text{logit}(p_0(\mathbf{x})) + \eta(\mathbf{x}) \cdot t) - p_0(\mathbf{x}).
% \end{equation}
% This assumption captures the realistic diminishing marginal returns of coupon interventions while maintaining a tractable and optimizable parameter $\eta(\mathbf{x})$ (unit sensitivity) to guide cost-constrained treatment allocation.
\end{assumption}

% To accurately assess the impact of interventions on CTR differences, we model the problem as a causal inference task and make the following assumptions:

% \begin{assumption}[Consistency, unconfoundedness and overlap]

% We adopt the following standard assumptions from causal inference theory:

% \begin{itemize}[leftmargin=8pt]
% \item \textbf{Consistency}: If individual $i$ is assigned treatment $w_i$, the observed outcome $Y_i$ and the potential outcomes $Y_i(w_i)$ are consistent.
% \item \textbf{Unconfoundedness}: there are no unobserved confounders, so that $Y \perp T |X$.
% \item \textbf{Overlap}: For each observation $i$, only one of the latent outcomes can be observed, i.e. $0 < \pi(x) < 1, \forall x \in X$ .
% \end{itemize}
  
% \end{assumption}

% \begin{assumption}[Monotonic Treatment Effect]
% Given the existence of intervention level $d_i$, it is assumed that the individual causal effect of the intervention exhibits a monotonically increasing relationship with the intervention level $d_i$. Formally, for any instance $i$, if $d_a < d_b$, it follows that 
% \begin{equation}
%   \text{CATE}(x_i;d_a) < \text{CATE}(x_i;d_b).
% \end{equation}

% This implies that as the intervention level increases, the causal effect of the intervention also intensifies.
% \end{assumption}

\subsection{Problem Formulation}

Our primary objective is to disentangle the intrinsic user preference (Base CTR) from the intervention effect within the observed data. Building upon this, we aim to accurately estimate the Conditional Average Treatment Effect (CATE) across continuous treatment intensities. We frame this problem under the Neyman-Rubin Potential Outcome framework~\cite{Rubin01032005}.

% \begin{definition}[\textbf{Base \& Treated CTR}]
% For a user instance with feature vector $\mathbf{x}$, let $Y(0)$ and $Y(t)$ denote the potential outcomes under the control state ($w=0$) and intervention intensity $t$ ($w=1$), respectively. We define the \textbf{Base CTR}, denoted as $p_0(\mathbf{x})$, and the \textbf{Treated CTR}, denoted as $p_t(\mathbf{x}, t)$, as the expected potential outcomes conditional on the covariates:
% \begin{align}
% p_0(\mathbf{x}) &= \mathbb{E}[Y(0) \mid X=\mathbf{x}], \\
% p_t(\mathbf{x}, t) &= \mathbb{E}[Y(t) \mid X=\mathbf{x}, T=t].
% \end{align}
% \end{definition}

\begin{definition}[\textbf{Base \& Treated CTR}]\label{def:ctr}
For a user instance with feature vector $\mathbf{x}$, let $Y(t)$ denote the potential outcome when the user is exposed to an intervention intensity $t \in \mathcal{T}$, where $t=0$ specifically denotes the control state (no intervention). We define the \textbf{Base CTR}, denoted as $p_0(\mathbf{x})$, and the \textbf{Treated CTR}, denoted as $p_t(\mathbf{x}, t)$, as the expected potential outcomes conditional strictly on the user covariates:
\begin{align}
p_0(\mathbf{x}) &\triangleq \mathbb{E}[Y(0) \mid \mathbf{x}], \\
p_t(\mathbf{x}, t) &\triangleq \mathbb{E}[Y(t) \mid \mathbf{x}].
\end{align}
\end{definition}

% \begin{definition}[\textbf{Unit CATE}]\label{Unit CATE}
% The individual treatment effect is defined as the difference between the potential outcomes $Y(t)$ and $Y(0)$. Drawing on Assumption \ref{assump:monotonicity}, we decompose the Conditional Average Treatment Effect (CATE) for a specific intensity $t$, denoted as $\tau(\mathbf{x}, t)$, into the product of the dosage magnitude and a Unit CATE (uCATE), denoted as $\eta(\mathbf{x})$:

% \begin{equation}
% \tau(\mathbf{x}, t) = \text{logit}(p_t(\mathbf{x}, t)) - \text{logit}(p_0(\mathbf{x})) = t \cdot \eta(\mathbf{x}). \label{eq:cate_def}
% \end{equation}

% Here, $\eta(\mathbf{x})$ represents the intrinsic marginal sensitivity of user $\mathbf{x}$ to the intervention (i.e., the uplift generated per unit of dosage), while $t$ represents the intervention intensity.
% \end{definition}

% \begin{definition}[\textbf{Unit CATE}]\label{Unit CATE}
\begin{definition}[\textbf{Latent Unit Sensitivity (uCATE)}]\label{def:latent_ucate}
Building upon Assumption \ref{assump:latent_linearity}, we formally define the scalar parameter $\eta(\mathbf{x})$ as the Latent Unit CATE (uCATE). It represents the intrinsic marginal sensitivity of user $\mathbf{x}$ to the intervention in the log-odds space. Consequently, the standard Conditional Average Treatment Effect (CATE) in the observable probability space, defined as the absolute difference between potential outcomes, can be systematically parameterized by this intensity-invariant uCATE:
\begin{align}
\tau(\mathbf{x}, t) &\triangleq p_t(\mathbf{x}, t) - p_0(\mathbf{x}) \nonumber\\ &= \sigma\Big(\text{logit}(p_0(\mathbf{x})) + \eta(\mathbf{x}) \cdot t\Big) - p_0(\mathbf{x}). \label{eq:cate_derivation} 
\end{align}
Through this definition, the standard CATE $\tau(\mathbf{x}, t)$ inherently captures the non-linear saturation effect as the intensity $t$ increases, while maintaining a tractable and optimizable scalar $\eta(\mathbf{x})$ for downstream cost-constrained allocation.\end{definition}

% From Eq. \eqref{eq:cate_def}, the observed outcome for a treated sample can be decomposed into the base probability and the incremental lift:

% \begin{equation}
% p_d(\mathbf{x}, t) = p_0(\mathbf{x}) + \tau(\mathbf{x}, t).
% \end{equation}

% To ensure identifying validity, we construct our model based on two key premises regarding the relationship between $p_0$ and $p_t$:

% \begin{itemize}[leftmargin=8pt]
% \item \textbf{Structural Correlation}: The distributions of the Base CTR $p_0(\mathbf{x})$ and the Treated CTR $p_t(\mathbf{x}, d)$ are inherently correlated, as they share the same underlying user representations. In the hypothetical absence of a treatment effect (i.e., $\tau(\mathbf{x}, t) \to 0$), the two distributions should asymptotically align.
% \item \textbf{Effect Heterogeneity}: The impact of interventions is non-uniform across the population. We assume the existence of heterogeneous treatment effects, where the incremental uplift $\tau(\mathbf{x}, t)$ varies systematically based on user context $\mathbf{x}$ and intensity value $t$.
% \end{itemize}

% \subsection{How does CTR bias manifest due to its lack of awareness regarding coupon issuance?}

% (To be supplemented)

\section{Framework}

% In this section, we formally present the proposed Unified Multi-Valued Treatment Network (UniMVT), detailing its architectural components and the associated learning paradigm.

\subsection{Overview}
As illustrated in Figure \ref{fig:main}, our framework comprises two modules. The first is a Deconfounded Causal Representation (DCR) Layer, which employs a Mixture-of-Experts (MoE) structure to encode input covariates and explicitly disentangle them into treatment-sensitive representations and treatment-invariant confounding representations. The second is a Heterogeneous Treatment Effect (HTE) Network, which leverages these disentangled embeddings to jointly estimate the baseline CTR $p_0(\mathbf{x})$ and the latent unit sensitivity $\eta(\mathbf{x})$. Guided by Definition \ref{def:latent_ucate}, the intensity-dependent uplift $\tau(\mathbf{x}, t)$ is subsequently derived analytically for any specific intervention intensity $t$. 
% Training utilizes all observational data.
The training strategy uses all observational data, including both coupon and no-coupon populations, to facilitate learning across the full covariate-treatment space.
% As illustrated in Figure \ref{fig:main}, our framework comprises two modules. The Deconfounded Causal Representation (DCR) Layer uses MoE to disentangle covariates into treatment-sensitive and confounding representations. Using these, the Heterogeneous Treatment Effect (HTE) Network jointly estimates baseline CTR $p_0(\mathbf{x})$ and latent sensitivity $\eta(\mathbf{x})$, analytically deriving intensity-dependent uplift $\tau(\mathbf{x}, t)$ (Definition \ref{def:latent_ucate}). Training utilizes all observational data.

\begin{figure*}[ht]
  \centering
  \includegraphics[width=\linewidth]{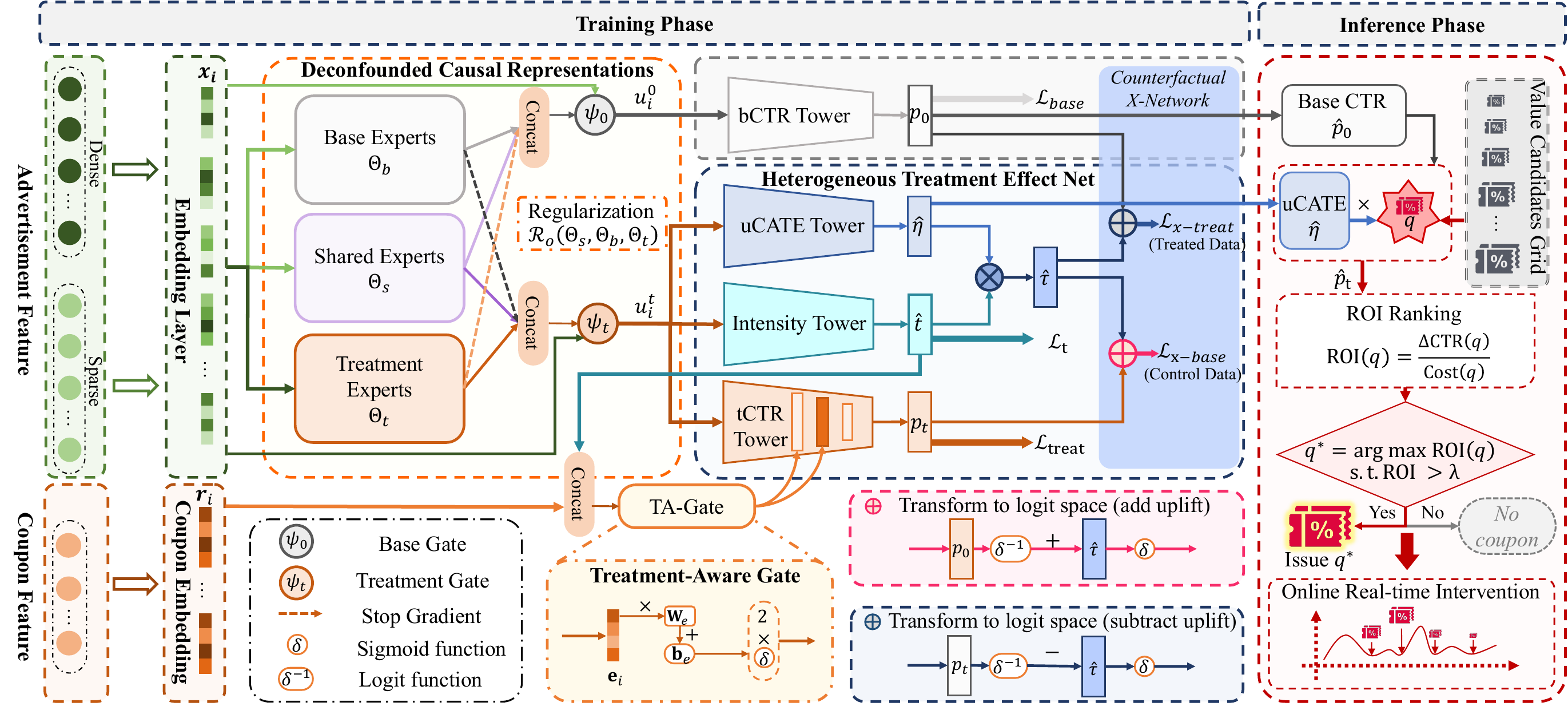}
  \caption{Overview of UniMVT, consisting of: (i) a DCR layer that disentangles treatment-sensitive and confounding features via MoE; and (ii) an HTE network for jointly estimating base CTR, latent uCATE, and intensity-dependent uplift.}
  \label{fig:main}
  \Description{Architecture of UniMVT, showing the DCR layer, HTE network, training losses, counterfactual paths, and online inference flow.}
  \vspace{-1em}
\end{figure*}

\subsection{Deconfounded Causal Representation}
The Deconfounded Causal Representation (DCR) module employs a Multi-gate Mixture-of-Experts (MMoE) architecture~\cite{MMOE2018} to encode input features $\mathbf{x}_i$ into disentangled latent embeddings. We instantiate three distinct groups of MLP-based experts: \textit{Base} ($\mathbf{H}_i^b$), \textit{Shared} ($\mathbf{H}_i^s$), and \textit{Treatment} ($\mathbf{H}_i^t$) experts, which capture global baseline preferences, treatment-invariant confounders, and intervention-sensitive features, respectively.

To synthesize these representations, task-specific gating networks ($\mathbf{g}_{i}^0$ and $\mathbf{g}_{i}^t$) generate soft attention weights. We construct the final disentangled embeddings, $\mathbf{u}_{i}^0$ and $\mathbf{u}_{i}^t$, using a Stop-Gradient ($\mathsf{SG}$) operator to prevent gradient leakage between conflicting objectives while allowing information flow through shared confounders:
\begin{align}
\mathbf{u}_{i}^0 &= (\mathbf{g}_{i}^0)^\top \mathsf{CONCAT}\Big(\mathbf{H}_{i}^b, \mathbf{H}_{i}^s, \mathsf{SG}(\mathbf{H}_i^t)\Big), \\ 
\mathbf{u}_{i}^t &= (\mathbf{g}_{i}^t)^\top \mathsf{CONCAT}\Big(\mathsf{SG}(\mathbf{H}_{i}^b), \mathbf{H}_{i}^s, \mathbf{H}_i^t\Big). 
\end{align} 
Consistent with Definition \ref{def:latent_ucate}, $\mathbf{u}_{i}^0$ is fed into the Base Tower to estimate the counterfactual baseline $p_0(\mathbf{x}_i)$, while $\mathbf{u}_{i}^t$ enters the Heterogeneous Treatment Effect Net (HTENet) to directly estimate the Latent Unit CATE $\eta(\mathbf{x}_i)$.

Furthermore, to explicitly enforce the disentanglement of confounding factors from treatment-specific heterogeneity, we impose an orthogonality penalty~\cite{curth2021inductive} across all distinctive expert pairs:
\begin{equation}
\mathcal{R}_{orth} = \sum_{l=1}^L \sum_{(u, v) \in \mathcal{P}} \left\| (\mathbf{\Theta}^l_u)^\top \mathbf{\Theta}^l_v \right\|_F^2,
\end{equation}
where $\mathbf{\Theta}^l_{(\cdot)}$ denotes the $l$-th layer weight matrix of the respective expert group, $\mathcal{P} = \{(b, s), (t, s), (b, t)\}$, and $\|\cdot\|_F^2$ is the squared Frobenius norm. This constraint minimizes the correlation between the parameter spaces of different expert groups, encouraging the extraction of strictly non-overlapping causal representations.

\subsection{Heterogeneous Treatment Effect Net}

To prevent task interference and negative transfer between factual and counterfactual estimations, we employ a decoupled tower architecture. We formally define the Heterogeneous Treatment Effect (HTE) Network, consisting of two distinct prediction pathways and a specific mechanism for injecting treatment information.

\subsubsection{Decoupled Prediction Towers}
To isolate the baseline propensity from intervention-specific dynamics, we instantiate two independent MLP towers:

\begin{itemize}[leftmargin=8pt]
\item \textbf{Base Tower}: Ingests the baseline-specific representation $\mathbf{u}_i^0$ from the DCR layer to estimate the intrinsic click probability $\hat{p}_0(\mathbf{x}_i)$. 
\item \textbf{Treatment Tower}: Consumes the treatment-aware representation $\mathbf{u}_i^t$ and the treatment attribute embedding $\mathbf{e}^t_i$ to predict the intervention-conditioned outcome $\hat{p}_t(\mathbf{x}_i, t_i)$. 
\end{itemize}

To mitigate the signal dilution inherent in simple feature concatenation, we introduce a Treatment-Aware Gate (TA-Gate), adapting the LHUC framework~\cite{ppnet2023} for continuous interventions. For each hidden layer $l$, we compute a context-dependent scaling vector $\mathbf{a}_i^{(l)} \in (0, 2)$ derived from the treatment embedding $\mathbf{e}^t_i$:

\begin{equation}
\mathbf{a}_i^{(l)} = 2 \cdot \sigma(\mathbf{W}_g^{(l)} \mathbf{e}^t_i + \mathbf{b}_g^{(l)}), \quad \tilde{\mathbf{h}}_i^{(l)} = \mathbf{a}_i^{(l)} \odot \mathbf{h}_i^{(l)}.
\end{equation}

Here, $\sigma(\cdot)$ is the sigmoid function and $\odot$ denotes the element-wise product. This mechanism ensures that the treatment intensity dynamically modulates the feature processing pathway at multiple abstraction levels.

We formulate the optimization objective using the binary cross-entropy loss function, denoted as $\ell_{\text{BCE}}(\cdot, \cdot)$. Let $\mathcal{C} = \{i \mid t_i=0\}$ and $\mathcal{T} = \{i \mid t_i > 0\}$ denote the control and treatment groups, respectively. The factual loss functions for the two towers are defined as:

\begin{align}
\mathcal{L}_{base} &= \sum_{i \in \mathcal{C}} \ell_{\text{BCE}}(y_i, \hat{p}_0(\mathbf{x}_i)), \\ 
\mathcal{L}_{treat} &= \sum_{i \in \mathcal{T}} \ell_{\text{BCE}}(y_i, \hat{p}_t(\mathbf{x}_i, t_i)).
\end{align}

\subsubsection{Counterfactual X-Network}
 
Inspired by the X-learner~\cite{kunzel2019metalearners}, we instantiate an X-Network as a bridging module to enable counterfactual reasoning. It shares the bottom representation $\mathbf{u}_{i}^t$ to infer hypothetical outcomes under alternative treatment regimes.

To operationalize this for the control group ($\mathcal{C}$), where actual intervention intensity is unobserved, we introduce an auxiliary Intensity Prediction Head to impute a feasible hypothetical dosage $\hat{t}_i$. We strictly enforce domain constraints by projecting the latent representation into a physical range $[t_{\min}, t_{\max}]$:

\begin{align}
\hat{t}_i &= \sigma(\mathsf{MLP}(\mathsf{SG}(\mathbf{u}^t_i))) \cdot (t_{\max} - t_{\min}) + t_{\min}, \\ 
\mathcal{L}_{t} &= \sum_{i \in \mathcal{T}} \left( \lambda_{1} ( t_i - \hat{t}_i )^2 + \lambda_{2} | t_i - \hat{t}_i | \right),
\end{align}

where $\sigma(\cdot)$ ensures the normalized prediction lies strictly within $(0, 1)$. This head is trained exclusively on the treated group $\mathcal{T}$ and applied to infer dosages for the control group. 

Simultaneously, we estimate the intensity-invariant Latent Unit CATE ($\hat{\eta}_i$) via the treatment representation. We employ a ReLU activation to structurally enforce non-negativity, satisfying the monotonicity assumption:
\begin{equation}
\hat{\eta}_i = \text{ReLU}(\mathsf{MLP}(\mathbf{u}^t_i)).
\end{equation}

Guided by Definition \ref{def:latent_ucate}, we formulate the cross-network estimators by injecting or removing the latent treatment shift (dosage $\times$ sensitivity) in the logit space. Let $\sigma^{-1}(\cdot)$ denote the logit function. The counterfactual estimations are tailored to the respective groups:

\begin{align}
&\textbf{For } i \in \mathcal{T}: \quad \hat{p}_{t}'(\mathbf{x}_i; t_i) = \sigma\Big( \sigma^{-1}(\hat{p}_{0}(\mathbf{x}_i)) + t_i \cdot \hat{\eta}_i \Big), \label{eq:xtreat} \\ 
&\textbf{For } i \in \mathcal{C}: \quad \hat{p}_{0}'(\mathbf{x}_i) = \sigma\Big( \sigma^{-1}(\hat{p}_{t}(\mathbf{x}_i, \hat{t}_i)) - \hat{t}_i \cdot \hat{\eta}_i \Big). \label{eq:xbase}
\end{align}

Equation \ref{eq:xtreat} reconstructs the treated outcome by augmenting the base prediction with the \textit{actual} treatment shift $t_i \cdot \hat{\eta}_i$. Conversely, Equation \ref{eq:xbase} reconstructs the baseline for control users by stripping the \textit{imputed} treatment shift $\hat{t}_i \cdot \hat{\eta}_i$ from their hypothetical treated prediction.

To enforce consistency between these counterfactual deductions and the factual observed labels, we employ Mean Squared Error (MSE) regularization:

\begin{align}
\mathcal{L}_{\text{x-treat}} &= \sum_{i \in \mathcal{T}} ( y_i - \hat{p}_{t}'(\mathbf{x}_i; t_i) )^2, \\
\mathcal{L}_{\text{x-base}} &= \sum_{i \in \mathcal{C}} ( y_i - \hat{p}_{0}'(\mathbf{x}_i) )^2.
\end{align}

This mutual regularization propagates supervision from treated to control (and vice versa) via the $\hat{\eta}_i$ parameter, mitigating selection bias and implicitly calibrating the latent sensitivity. The combined regularization loss is $\mathcal{L}_X = \mathcal{L}_{\text{x-treat}} + \mathcal{L}_{\text{x-base}}$.

The full HTENet is optimized end-to-end via the joint loss:
\begin{equation} 
\mathcal{L} = \lambda_{base} \mathcal{L}_{base} + \lambda_{treat} \mathcal{L}_{treat} + \lambda_{t} \mathcal{L}_{t} + \lambda_{X} \mathcal{L}_{X} + \lambda_{o} \mathcal{R}_o, 
\end{equation}
where $\lambda_{(\cdot)}$ are hyperparameters balancing factual accuracy, intensity imputation, counterfactual consistency, and orthogonality.

\subsection{Theoretical Analysis}
To theoretically justify the effectiveness of our proposed framework, we analyze the convergence properties of the unit uplift estimator. We show that under mild assumptions regarding the dosage distribution and the consistency of base estimators, the proposed counterfactual calibration loss guarantees the identification of the true unit uplift.

\begin{assumption}[Non-trivial Intensity and Common Support]\label{intensity}
% The treatment intensity $t$ is bounded away from zero, i.e., there exists a constant $t_{\min} > 0$ such that $|t| \geq t_{\min}$ almost surely for all valid interventions. Furthermore, we assume distributional consistency: the conditional distribution of feasible intensities $P(t \mid \mathbf{x})$ depends solely on the feature set $\mathbf{x}$ (e.g., item price, marketing strategy). This implies that the dosage mechanism learned from the treated population $\mathcal{T}$ generalizes to the full sample space $\mathcal{S} = \mathcal{T} \cup \mathcal{C}$. Formally, for any user-item pair $\mathbf{x}$, the support of the potential intensity distribution in the control group is contained within the support of the treated group:
The treatment intensity $t$ is bounded away from zero, i.e., there exists a constant $t_{min} > 0$ such that $|t|\ge t_{min}$ almost surely for all valid interventions. Furthermore, the conditional distribution of feasible intensities $P(t|x)$ depends solely on the feature set $x$. The support of the potential intensity distribution in the control group is contained within the support of the treated group:

\begin{equation}
\text{supp}(P(t \mid \mathbf{x}, i \in \mathcal{C})) \subseteq \text{supp}(P(t \mid \mathbf{x}, i \in \mathcal{T})).
\end{equation}
\end{assumption}

\begin{assumption}[Consistency of Base Estimators]\label{estimators}
The estimators for the treated outcome $\hat{y}_t$ and the intensity $\hat{t}$ are consistent, meaning that as the sample size $N \to \infty$, $\hat{p}_t \xrightarrow{p} p_t$ and $\hat{t} \xrightarrow{p} t$.
\end{assumption}

Based on these assumptions, we establish the following theorem regarding the convergence of the unit uplift estimator $\hat{\eta}$.

\begin{theorem}[Convergence of Unit Uplift]\label{Convergence of Unit Uplift}
Under Assumptions \ref{intensity} and \ref{estimators}, minimizing the counterfactual calibration loss $\mathcal{L}_{x-treat}$ and $\mathcal{L}_{x-base}$ ensures that the estimated unit uplift $\hat{\eta}$ converges in probability to the true unit uplift $\eta$.
\end{theorem}

\noindent \textbf{Proof Sketch}. We take  $\mathcal{L}_{x-treat}$ as an example; $\mathcal{L}_{x-base}$ can be proven in the same manner. Recall the Assumption\ref{assump:latent_linearity} that the counterfactual calibration loss for the control group enforces the constraint: ${p}_0(\mathbf{x}_i) + {\eta} \cdot {t} \approx {p}_t(\mathbf{x}_i, t_i)$.

Let $\Delta = |\hat{\eta} - \eta|$ be the objective to be minimized. Substituting the true ${p}_t(\mathbf{x}_i, t_i)$, we analyze the error term:

\begin{align}
    \mathcal{L}_{\text{x-treat}} &= |(\hat{p}_0(\mathbf{x}_i) + \hat{\eta}\hat{t}) - ({p}_0(\mathbf{x}_i) + {\eta} \cdot {t})|
\end{align}

By the triangle inequality and adding the cross-term $\hat{\eta}{t}$, we have:

\begin{align} 
\mathcal{L}_{\text{x-treat}} &= |(\hat{p}_0(\mathbf{x}_i) - {p}_0(\mathbf{x}_i)) + (\hat{\eta}\hat{t} - \eta t)| \\ 
&= |(\hat{p}_0(\mathbf{x}_i) - {p}_0(\mathbf{x}_i)) + \hat{\eta}(\hat{t} - t) + {t}(\hat{\eta} - \eta)| 
\end{align}
Minimizing $\mathcal{L}_{\text{x-treat}}$ implies $\Delta \to 0$.

By Assumption \ref{estimators}, the terms $|\hat{p}_0(\mathbf{x}_i) - {p}_0(\mathbf{x}_i)|$ and $|\hat{t} - t|$ converge to 0. Thus, the convergence of the loss relies on the term ${t}(\hat{\eta} - \eta) \to 0$. By Assumption \ref{intensity}, since $|\hat{t}| \to |t| \ge t_{\min} > 0$, the intensity is non-degenerate. Therefore, the only solution for the term to vanish is $\hat{\eta} - \eta \to 0$.
This completes the proof that $\hat{\eta} \xrightarrow{p} \eta$. $\square$
% For a more detailed proof, see the Appendix \ref{appendix_proof}.

% \subsection{Inference}

% At inference, UniMVT directly outputs the unit uplift $\hat{\eta}_i$ and estimated base CTR $\hat{p}_0(x_i) $. To guide coupon decisions, we simulate outcomes over a discrete grid of candidate coupon values $q \in [q_{\min}, q_{\max}]$:

% \begin{align}
% \Delta CTR = \hat{p}_{t}(x;q) - \hat{p}_{0}(x;q) = \sigma(\sigma^{-1}(\hat{p}_0(x_i)) + \hat{\eta}_i \cdot q) - \hat{p}_{0}(x;q).
% \end{align}

% For each $q$, compute the expected uplift-to-cost ratio (e.g., incremental CTR gain per unit cost). If this exceeds a predefined threshold, issue the voucher at intensity $q^*$; otherwise, the voucher should be withheld. This approach enables real-time, personalized intervention leveraging offline-learned causal effects.

\subsection{Inference}

At inference, UniMVT directly outputs the estimated unit sensitivity $\hat{\eta}_i$ and the base conversion probability $\hat{p}_0(x_i)$ for each user $i$. To guide coupon allocation decisions, we simulate outcomes over a discrete grid of candidate coupon values $q \in \{q_1, q_2, \dots, q_{\max}\}$.

First, we calculate the incremental conversion probability (Uplift) parameterized by $q$:

\begin{align}
\Delta \text{CTR}_i(q) &= \hat{p}_{t}(x_i;q) - \hat{p}_{0}(x_i) \nonumber \\
&= \sigma(\sigma^{-1}(\hat{p}_0(x_i)) + \hat{\eta}_i \cdot q) - \hat{p}_{0}(x_i) 
\end{align}

Then, we formulate a non-linear marginal ROI objective. By incorporating the saturation effect of the sigmoid function and the expected redemption probability $\gamma$, the expected cost inherently acts as a non-linear regularizer against the coupon face value:

\begin{align}
\text{ROI}_i(q) = \frac{\Delta \text{CTR}_i(q)}{\mathbb{E}[\text{Cost}_i(q)]} = \frac{\hat{p}_{t}(x_i;q) - \hat{p}_{0}(x_i)}{q \cdot \hat{p}_{t}(x_i;q) \cdot \gamma}
\end{align}

The optimal intervention $q^*$ is determined by evaluating $q^* = \arg\max_{q} \text{ROI}_i(q)$. If the optimal $\text{ROI}_i(q^*)$ exceeds a system-defined threshold $\lambda$, the voucher $q^*$ is issued; otherwise, the treatment is withheld ($q=0$).

\section{Experiment}

We evaluate three questions: \textbf{RQ1}, how UniMVT compares with the baselines; \textbf{RQ2}, how each component contributes; and \textbf{RQ3}, how UniMVT performs in online deployment.

\subsection{Evaluation Setup}

\subsubsection{\textbf{Dataset}}
% We evaluate our approach on both synthetic and realworld datasets to prove its effectiveness. Public datasets are excluded due to misalignment with the online platform's requirements, limiting their suitability for industrial-scale RS evaluation. %  compared to state-of-the-art methods
We evaluate our approach on both synthetic and real-world datasets to demonstrate its efficacy. 

\textbf{\textit{Synthetic Datasets}}: To evaluate the model against known causal ground truth, we synthesize three controlled benchmarks (\texttt{Syn-1}, \texttt{Syn-2}, and \texttt{Syn-3}) that closely mimic operational marketing environments. Each dataset contains 80,000 training samples and 8,000 test samples, with the test set consisting exclusively of unbiased Randomized Controlled Trial (RCT) samples.

Two design choices ensure industrial relevance. First, the individualized treatment-response curves strictly follow Assumption \ref{assump:latent_linearity}. Injecting monotonic linear shifts in logit space preserves probability bounds and produces the nonlinear saturation observed in practice. Second, treatment intensity $T$ follows unimodal distributions in \texttt{Syn-1}/\texttt{Syn-2} and a multimodal distribution in \texttt{Syn-3}, reflecting production coupon tiers. We introduce realistic selection bias by confounding $T$ with user covariates $\mathbf{X}$, so high-value users have a higher propensity to receive high-intensity interventions.

\textbf{\textit{Real-world Dataset}}: We curate an offline set from Kuaishou live-stream coupon traffic, including both non-intervened organic interactions and interactions subject to heterogeneous intervention intensities. Inputs encompass user profiles, live-stream-room context, and coupon descriptors such as intensity and style, yielding a production-scale CTR task with multi-valued interventions.

\subsubsection{\textbf{Baselines}}
% We selected representative causal inference methods for both binary and multi-value interventions:
% S-Learner, T-Learner, FlexTENet, CFRNet, DRNet, VCNet, CEVAE, DESCN, EFIN, RMNet and XTNet.
% Several baselines were originally designed for binary treatments. We extend them to multi-valued interventions by partitioning the intensity space into intervals and attaching separate prediction heads (on a shared backbone) for each interval.
% We compare UniMVT with leading causal inference methods: S-Learner~\cite{kunzel2019metalearners}, T-Learner~\cite{kunzel2019metalearners}, FlexTENet~\cite{curth2021inductive}, CFRNet~\cite{shalit2017estimating}, DRNet~\cite{schwab2020learning}, VCNet~\cite{nie2021vcnet}, CEVAE~\cite{louizos2017causal}, and DESCN~\cite{zhong2022descn}. For binary-treatment baselines, we extend their applicability to the multi-valued setting by treating the intervention dose $d$ as a distinct input feature. Their outputs are subsequently normalized by $d$ to produce comparable unit uplift estimates ($\hat{\tau}_{unit} = \hat{\tau}_{total} / d$), ensuring a consistent evaluation landscape across all multi-valued scenarios.
We compare UniMVT with S-Learner and T-Learner~\cite{kunzel2019metalearners}, FlexTENet~\cite{curth2021inductive}, CFRNet~\cite{shalit2017estimating}, DRNet~\cite{schwab2020learning}, VCNet~\cite{nie2021vcnet}, and DESCN~\cite{zhong2022descn}. We extend binary-treatment baselines by using intensity $t$ as an input; for models predicting full uplift, setting $t=1$ yields comparable unit uplift, $\hat{\tau}_{unit}=\text{Model}(X,t=1)$.

\subsubsection{\textbf{Evaluation Metrics}}
We use separate metrics for CTR debiasing and uplift ranking.

\textbf{\textit{Basic CTR Prediction:}} AUC and LogLoss measure ranking and calibration of the counterfactual conversion probability under control (no treatment).

\textbf{\textit{Multi-valued Uplift Estimation:}} Binary uplift metrics are insufficient for continuous treatments, so we use a \textit{Cumulative Slope (CS)} framework. For a test set $\mathcal{O}$ of size $N$, samples are ranked by predicted Latent Unit CATE $\hat{\eta}(\mathbf{x})$. Let $\mathcal{O}_\phi$ be the top $\phi$ fraction, $\phi\in[0,1]$, and $\hat{\beta}(\phi)$ the regression slope of outcome $Y$ on intensity $T$ within $\mathcal{O}_\phi$.

\begin{definition}[CS-AUUC]
The {Cumulative Slope Area Under the Uplift Curve (CS-AUUC)} measures rank-ordered sensitivity as the area under the realized slope $\hat{\beta}(\phi)$ scaled by bucket size:
\begin{align}
    \text{CS-AUUC} = \int_{0}^{1} \hat{\beta}(\phi) \cdot (\phi N) \, d\phi.
\end{align}
Higher CS-AUUC means more intensity-sensitive users are ranked first.
\end{definition}

\begin{definition}[CS-Qini]
The \textbf{Cumulative Slope Qini (CS-Qini)} measures ranking gain over random targeting. With $\beta_{\text{global}}$ denoting the average slope on $\mathcal{O}$,
\begin{align}
    \text{CS-Qini} = \int_{0}^{1} \left( \hat{\beta}(\phi) - \beta_{\text{global}} \right) \cdot (\phi N) \, d\phi.
\end{align}
Thus CS-Qini quantifies the incremental gain of personalized over random allocation.
\end{definition}

\subsection{Main Results (RQ1)}
% \begin{figure*}[t]
%   \centering
%   \includegraphics[width=\linewidth]{data_pics/kdd_combined_plots.pdf}
%   \caption{CS-AUUC comparison of different methods across three synthetic datasets. All curves start from the origin (0,0) and are compared against the Global Slope.} 
%   \label{fig:CS-AUUC_comparison}
% \end{figure*}
\begin{figure}[t]
  \centering
  \includegraphics[width=\linewidth]{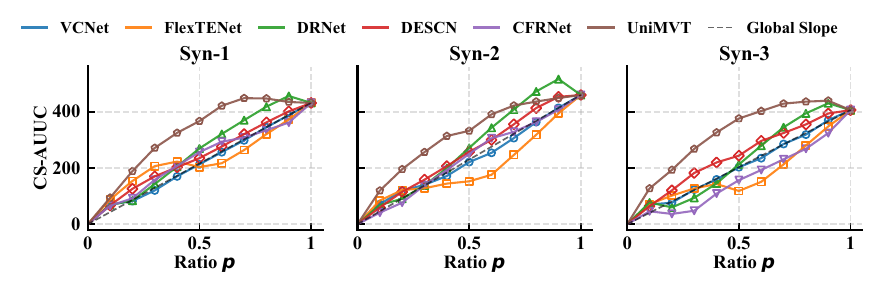}
  \caption{CS-AUUC comparison of different methods across three synthetic datasets. All curves start from the origin (0,0) and are compared against the Global Slope.} 
  \label{fig:CS-AUUC_comparison}
  \Description{Three line charts comparing CS-AUUC curves across methods on the Syn-1, Syn-2, and Syn-3 datasets.}
  \vspace{-0.6em}
\end{figure}
% Table \ref{tab:performance} reports the comparative performance of our proposed method, UniMVT, against state-of-the-art baselines on the Syn-1 dataset. The results demonstrate that UniMVT consistently achieves the best performance across both counterfactual CTR prediction and uplift estimation metrics.
% \begin{figure*}[htbp]
%     \centering
%     \subfigure[Scenario 1\label{fig:syn1}]{
%         \centering
%         \includegraphics[width=0.3\linewidth]{data_pics/method_comparison_syn_1.pdf}}
%     \hspace{0.01\linewidth}
%     \subfigure[Scenario 2\label{fig:syn2}]{
%         \centering
%         \includegraphics[width=0.3\linewidth]{data_pics/method_comparison_syn_2.pdf}}
%     \hspace{0.01\linewidth}
%     \subfigure[Scenario 3\label{fig:syn3}]{
%         \centering
%         \includegraphics[width=0.3\linewidth]{data_pics/method_comparison_syn_3.pdf}}
%     \caption{Comparison of different methods across three synthetic datasets. All curves start from the origin (0,0) and are compared against the Global Slope.}
%     \label{fig:overall_comparison}
% \end{figure*}
Table \ref{tab:overall_performance} presents a comprehensive comparison between UniMVT and seven baselines across synthetic and production datasets.

\textbf{Base CTR Evaluation:} The T-Learner's superiority over the S-Learner indicates that decoupled pathways are crucial for isolating base conversion propensity. Building on this, UniMVT achieves highly competitive, often state-of-the-art AUC and LogLoss. This confirms that our DCR layer successfully disentangles treatment-invariant representations from intervention signals.

\textbf{Uplift Evaluation:} UniMVT establishes a dominant advantage. While baselines like DRNet perform adequately on simple unimodal data, their efficacy degrades significantly in complex multimodal (\texttt{Syn-3}) and industrial (\texttt{Prod}) settings, often yielding near-random CS-Qini scores (Figure \ref{fig:CS-AUUC_comparison}). This suggests that traditional models are excessively confounded by the absolute treatment magnitude $t$, assigning higher scores to larger coupon values rather than capturing intrinsic user sensitivity.

Conversely, UniMVT consistently achieves the highest CS-AUUC and CS-Qini across all scenarios. By explicitly modeling the Latent Unit CATE $\eta(\mathbf{x})$ and leveraging X-Network regularization, our framework accurately prioritizes users with genuine treatment sensitivity, demonstrating its robustness and scalability for real-world marketing interventions.

\begin{table}[t]
\centering
\caption{Model Performance Comparison on Synthetic and Production Datasets}
\label{tab:overall_performance}
\small
\footnotesize
\setlength{\tabcolsep}{3.5pt}
\begin{tabular}{ll|cc|cc}
\toprule
\multirow{2}{*}{\textbf{Method}} & \multirow{2}{*}{\textbf{Dataset}} & \multicolumn{2}{c|}{Base CTR Evaluation} & \multicolumn{2}{c}{Uplift Evaluation} \\ \cline{3-6}
 & & AUC $\uparrow$ & LogLoss $\downarrow$ & CS-AUUC $\uparrow$ & CS-QINI $\uparrow$ \\ \midrule

% --- T-Learner ---
\multirow{4}{*}{T-Learner~\cite{kunzel2019metalearners}} 
 & Syn-1 & 0.6918 & 0.3892 & 245.05 & 31.66 \\
 & Syn-2 & 0.6932 & 0.3855 & 247.27 & 19.93 \\
 & Syn-3 & 0.7067 & 0.3799 & \underline{257.37} & \underline{56.61} \\ 
 & Prod & 0.9109 & 0.2678 & - & 0.0177 \\
\midrule

% --- S-Learner ---
\multirow{4}{*}{S-Learner~\cite{kunzel2019metalearners}} 
 & Syn-1 & 0.6937 & 0.3879 & 219.84 & 6.45 \\
 & Syn-2 & 0.7032 & 0.3813 & 225.22 & -2.12 \\
 & Syn-3 & 0.7098 & 0.3784 & 219.03 & 18.27 \\ 
 & Prod & 0.9011 & 0.2855 & - & 0.0001 \\
\midrule

% --- CFRNet ---
\multirow{4}{*}{CFRNet~\cite{shalit2017estimating}} 
 & Syn-1 & 0.6936 & 0.3863 & 196.87 & -16.52 \\
 & Syn-2 & 0.6998 & 0.3810 & 231.26 & 3.92 \\
 & Syn-3 & 0.7101 & 0.3767 & 159.14 & -41.62 \\ 
 & Prod & 0.9097 & 0.2730 & - & \underline{0.0178} \\
\midrule

% --- FlexTENet ---
\multirow{4}{*}{FlexTENet~\cite{curth2021inductive}} 
 & Syn-1 & 0.6938 & 0.3868 & 221.84 & 8.45 \\
 & Syn-2 & 0.7026 & 0.3805 & 195.21 & -32.14 \\
 & Syn-3 & 0.7113 & 0.3757 & 171.99 & -28.77 \\ 
 & Prod & 0.8870 & 0.2997 & - & 0.0166 \\
\midrule

% --- DESCN ---
\multirow{4}{*}{DESCN~\cite{zhong2022descn}} 
 & Syn-1 & 0.6851 & 0.3910 & 235.92 & 22.52 \\
 & Syn-2 & 0.6958 & 0.3823 & 251.70 & 24.36 \\
 & Syn-3 & 0.7050 & 0.3785 & 236.96 & 36.20 \\ 
 & Prod & 0.9394 & 0.2161 & - & 0.0170 \\
\midrule

% --- DRNet ---
\multirow{4}{*}{DRNet~\cite{schwab2020learning}} 
 & Syn-1 & 0.6944 & 0.3864 & \underline{251.51} & \underline{38.12} \\
 & Syn-2 & \textbf{0.7028} & \textbf{0.3799} & \underline{270.12} & \underline{42.78} \\
 & Syn-3 & 0.7112 & 0.3759 & 194.31 & -6.45 \\ 
 & Prod & 0.9374 & 0.2155 & - & -0.0003 \\
 \midrule

% --- VCNet ---
\multirow{4}{*}{VCNet~\cite{nie2021vcnet}} 
 & Syn-1 & \underline{0.6942} & \underline{0.3862} & 212.26 & -1.14 \\
 & Syn-2 & 0.7018 & \underline{0.3801} & 225.31 & -2.04 \\
 & Syn-3 & \underline{0.7117} & \underline{0.3757} & 200.75 & -0.01 \\ 
 & Prod & \underline{0.9401} & \underline{0.2100} & - & 0.0005 \\
\midrule

% --- UniMVT (Ours) ---
\multirow{4}{*}{\textbf{UniMVT}} 
 & Syn-1 & \textbf{0.6945} & \textbf{0.3865} & \textbf{316.48} & \textbf{103.09} \\
 & Syn-2 & \underline{0.7025} & \underline{0.3801} & \textbf{311.29} & \textbf{83.94} \\
 & Syn-3 & \textbf{0.7122} & \textbf{0.3756} & \textbf{301.10} & \textbf{100.34} \\  
 & Prod & \textbf{0.9471} & \textbf{0.1776} & - & \textbf{0.0293} \\

\bottomrule
\end{tabular}
\vspace{-2em}
\end{table}
% \begin{table}[h]
% \centering
% \caption{Performance Comparison on Production Dataset}
% \label{tab:production_results}
% \small
% \begin{tabular}{l|cc|c}
% \toprule
% \textbf{Method} & \multicolumn{2}{c|}{\textbf{Base CTR Evaluation}} & \textbf{Uplift} \\ \cline{2-4}
% (Production) & AUC $\uparrow$ & LogLoss $\downarrow$ & CS-QINI $\uparrow$ \\ \midrule
% T-Learner~\cite{kunzel2019metalearners}   & 0.9109 & 0.2678 & 0.0177 \\
% S-Learner~\cite{kunzel2019metalearners}   & 0.9011 & 0.2855 & 0.0001 \\
% CFRNet~\cite{shalit2017estimating}      & 0.9097 & 0.2730 & \underline{0.0178} \\
% FlexTENet~\cite{curth2021inductive}    & 0.8870 & 0.2997 & 0.0166 \\
% DESCN~\cite{zhong2022descn}        & 0.9394 & 0.2161 & 0.0170 \\
% DRNet~\cite{schwab2020learning}        & 0.9374 & 0.2155 & -0.0003 \\
% VCNet~\cite{nie2021vcnet}        & \underline{0.9401} & \underline{0.2100} & 0.0005 \\ \midrule
% \textbf{UniMVT (Ours)} & \textbf{0.9471} & \textbf{0.1776} & \textbf{0.0293} \\
% \bottomrule
% \end{tabular}
% \end{table}

\subsection{Ablation Analysis (RQ2)}
% Table \ref{tab:ablation} validates the DCR and X-Network as essential, with their removal consistently degrading performance. Notably, omitting the auxiliary Treatment Tower preserves Base CTR AUC but yields a negative CS-Qini. This confirms that without explicit causal constraints, the model overfits observational biases, maximizing predictive accuracy at the expense of true uplift estimation.

Table \ref{tab:ablation} confirms that each component is necessary. Removing DCR degrades all metrics, while omitting X-Network collapses CS-Qini, validating their roles in disentanglement and counterfactual regularization. Removing linear monotonicity through non-logit shifts or unconstrained fitting also drives CS-Qini near zero despite high Base CTR AUC: without structural causal priors, deep networks optimize factual accuracy while overfitting observational bias.

% \begin{table}[t]
% \centering
% \caption{Ablation study of our UniMVT on production dataset.}
% \label{tab:ablation}
% \small
% \setlength{\tabcolsep}{6pt} % Adjusted for better readability
% \begin{tabular}{l|cc|c}
% \toprule
% \textbf{Metrics} & \multicolumn{2}{c|}{\textbf{Base CTR Evaluation}} & \textbf{Uplift Evaluation} \\ 
% \cmidrule(lr){2-3} \cmidrule(lr){4-4}
% & AUC $\uparrow$ & LogLoss $\downarrow$ & CS-QINI $\uparrow$ \\ \midrule

% w/o Logit Space                & 0.9457 & 0.1812 & 0.0008 \\
% w/o latent linear relationship                & {0.9401} & {0.1916} & 0.0004 \\
% w/o Orthogonal Regularization & {0.9444 } & {0.1868 } & 0.0042 \\
% w/o Shared Experts  & {0.9413 } & {0.2122} & -0.0072 \\ 
% w/o w/o couterfactul loss  & {0.9444} & {0.1868} & 0.0042 \\ \midrule
% \textbf{UniMVT}  & \textbf{0.9471} & \textbf{0.1776} & \textbf{0.0293} \\
% \bottomrule
% \end{tabular}
% \end{table}

\begin{table}[t]
\centering
\caption{Ablation study of UniMVT components on the production dataset.}
\label{tab:ablation}
\small
\footnotesize
\setlength{\tabcolsep}{3.5pt}
\begin{tabular}{l|cc|c}
\toprule
\multirow{2}{*}{\textbf{Model Variant}} & \multicolumn{2}{c|}{\textbf{Base CTR Evaluation}} & \textbf{Uplift Evaluation} \\ 
\cmidrule(lr){2-3} \cmidrule(lr){4-4}
& \textbf{AUC} $\uparrow$ & \textbf{LogLoss} $\downarrow$ & \textbf{CS-Qini} $\uparrow$ \\ \midrule

\multicolumn{4}{l}{\textit{Ablation on Linear Monotonicity}} \\
\quad w/o Logit Space                 & 0.9457 & 0.1812 & 0.0008 \\
\quad Unconstrained Treatment         & 0.9401 & 0.1916 & 0.0004 \\ \midrule

\multicolumn{4}{l}{\textit{Ablation on DCR Structure}} \\
\quad w/o Orthogonal Reg.             & 0.9460 & 0.1791 & 0.0097 \\
\quad w/o Shared Experts              & 0.9432 & 0.1801 & 0.0132 \\ \midrule

\multicolumn{4}{l}{\textit{Ablation on X-Network}} \\
\quad w/o Counterfactual Loss         & 0.9444 & 0.1868 & 0.0042 \\ \midrule

\textbf{UniMVT (Full Model)}          & \textbf{0.9471} & \textbf{0.1776} & \textbf{0.0293} \\
\bottomrule
\end{tabular}
\vspace{-0.5em}
\end{table}

\subsection{Online A/B Test (RQ3)}

% We evaluate our approach on a large-scale coupon intervention setting in the Kuaishou advertising platform-a short-video app with over 400 million daily active users.

% \subsubsection{System Overview}
% As Figure\ref{fig:engine}, in the serving engine, each app visit triggers a request that flows through recall and coarse ranking to select relevant items from an item pool. These candidates are fed to a fine-ranking CTR model. The model outputs both the base CTR and the per-user, per-item unit uplift attributable to issuing a coupon. A downstream strategy module then chooses the coupon with the maximum predicted revenue gain minus cost under platform constraints. Our method operates on the CTR model, improving the base CTR estimation and the uplift guidance simultaneously.

% \subsubsection{Online Experimental Results}
We conducted an online A/B test for UniMVT on a randomly sampled 10\% of Kuaishou's production traffic. Treatment and control buckets were mutually exclusive to prevent interference. We compare our approach against two established production baselines:

\begin{itemize}[leftmargin=8pt]
\item \textbf{Base (Two-Stage Model)}: The legacy production approach independently trains separate learners for base CTR and residual uplift. It suffers from error accumulation and fails to capture shared causal dynamics.
\item \textbf{Unified Model (Shared-Bottom)}: A multi-task framework uses a shared encoder with separate prediction towers to jointly estimate base CTR and uplift, representing a stronger representation-learning baseline.
\end{itemize}

% As for metrics, we report revenue in coupon and non-coupon scenarios.
% Table \ref{tab:ab} reports performance relative to the baseline. Compared to the Treatment-Agnostic Model, the proposed method delivers consistent improvements across all online metrics, indicating that jointly optimizing base click-through rate and intervention uplift yields higher-quality coupon allocation decisions. Relative to the T-Learner (Shared-Bottom), it achieves further gains, suggesting that full-space causal modeling enables more accurate CTR estimation and superior allocation quality. Following online A/B experiments, the system has been deployed to production for live-stream, direct-response coupon ad slots, with ongoing expansion to additional scenarios.

To evaluate online performance, we monitor business metrics---Revenue ($\Delta$Rev) and Return on Investment ($\Delta$ROI)---and calibration accuracy via Predicted Click-Over-Click Error ($\text{PCOC-Err}=|\text{PCOC}-1|$).

As shown in Table \ref{tab:ab}, UniMVT consistently outperforms the baselines. Compared with the Base Two-Stage model, it increases coupon revenue by 18.14\% and overall ROI by 8.80. Moreover, UniMVT reduces coupon PCOC-Err by 86.51\% and no-coupon PCOC-Err by 50.00\%. This substantial calibration improvement confirms that DCR and X-Network effectively isolate causal signals from confounding biases, enabling precise and granular coupon allocation.

Following the A/B test, UniMVT has been fully deployed in Kuaishou's live-stream direct-response ad slots, with ongoing expansion to other marketing scenarios.

% \begin{table}[htbp]
%   \centering
%   \caption{Online A/B Test Results}
%   \label{tab:ab}
%   \small
%   \begin{tabular}{lccc}
%     \toprule
%     \textbf{Model Settings} & \textbf{$\Delta$Rev (Coupon)}  & \textbf{$\Delta$Rev (No-Coupon)} \\  % & \textbf{Inc. ROI}
%     \midrule
%     Base & - & - \\ %  +1.874\% & 8.80 
%     Unified & +7.53\% & - \\ % & 8.80 
%     UniMVT & +18.14\% & +0.367\% \\ 
%     \bottomrule
%   \end{tabular}
% \end{table}

% \begin{table}[htbp]
%   \centering
%   \caption{Online A/B Test Results. Metrics denote relative improvements over the control group in each respective testing period. $\Delta$PCOC-Error denotes the relative reduction in calibration deviation ($|\text{PCOC} - 1|$).}
%   \label{tab:ab}
%   \footnotesize % 从 \small 进一步缩小到 \footnotesize
%   \setlength{\tabcolsep}{0.6mm} % 极限压缩列间距 (原为1.2mm)
%   \resizebox{\columnwidth}{!}{% 强制缩放到单栏宽度
%   \begin{tabular}{lccccc}
%     \toprule
%     \textbf{Model} & \makecell[c]{\textbf{$\Delta$Rev} \\ {Coupon}} & \makecell[c]{\textbf{$\Delta$Rev} \\ {No-Coupon}} & \textbf{$\Delta$ROI} & \makecell[c]{\textbf{$\Delta$PCOC-Err} \\ {Coupon}} & \makecell[c]{\textbf{$\Delta$PCOC-Err} \\ {No-Coupon}} \\ 
%     \midrule
%     Base & 0.00\% & 0.00\% & 0.00 & 0.00\% & 0.00\% \\ 
%     Unified Model & +7.53\% & +0.12\% & +3.53 & -70.66\% & -4.17\% \\ 
%     \midrule
%     \textbf{UniMVT} & \textbf{+18.14\%} & \textbf{+0.37\%} & +\textbf{8.80} & \textbf{-86.51\%} & \textbf{-50.00\%} \\ 
%     \bottomrule
%   \end{tabular}
%   }
% \end{table}

\begin{table}[htbp]
  \centering
  % \caption{Online A/B Test Results. Metrics denote relative improvements over the control group in each respective testing period. $\Delta$PCOC-Error denotes the relative reduction in calibration deviation ($|\text{PCOC} - 1|$).}
  \caption{Online A/B test results. Metrics report relative gains over the control per test period; $\Delta$PCOC-Error is the relative reduction in calibration deviation ($|\text{PCOC}-1|$).}
  \label{tab:ab}
  \resizebox{\columnwidth}{!}{%
  \begin{tabular}{l | ccc | cc}
    \toprule
    % 第一层：分类表头 (注意 \multicolumn 这里也需要加 | 来对齐下方的竖线)
    & \multicolumn{3}{c|}{\textbf{Business Metrics ($\uparrow$)}} 
    & \multicolumn{2}{c}{\textbf{Calibration Metrics ($\downarrow$)}} \\ 
    \cmidrule(lr){2-4} \cmidrule(lr){5-6}
    
    % 第二层：具体指标
    \textbf{Model} 
    & \makecell[c]{\textbf{$\Delta$Rev} \\ {Coupon}} 
    & \makecell[c]{\textbf{$\Delta$Rev} \\ {No-Coupon}} 
    & \textbf{$\Delta$ROI} 
    & \makecell[c]{\textbf{$\Delta$PCOC-Err} \\ {Coupon}} 
    & \makecell[c]{\textbf{$\Delta$PCOC-Err} \\ {No-Coupon}}  \\ 
    \midrule
    
    Base & 0.00\% & 0.00\% & +0.00 & 0.00 & 0.00\% \\ 
    Unified Model & +7.53\% & +0.12\% & +3.53 & -70.66\% & -4.17\% \\ 
    \midrule
    UniMVT & \textbf{+18.14\%} & \textbf{+0.37\%} & \textbf{+8.80} & \textbf{-86.51\%} & \textbf{-50.00\%} \\ 
    \bottomrule
  \end{tabular}%
  }
  \vspace{-0.5em}
\end{table}

% \begin{figure}[htbp]
%   \centering
%   \includegraphics[width=\linewidth]{pics/uplift_flow.pdf}
%   \caption{Illustration of the online pipeline.}
%   \label{fig:engine}
% \end{figure}

\section{Conclusion}
We proposed UniMVT to address confounding in multi-valued treatments by combining disentangled representation learning with full-space counterfactual inference. Offline and online results, supported by convergence guarantees, demonstrate improved debiased CTR estimation, heterogeneous uplift modeling, and marketing efficiency. Future work will explore more complex treatment-response curves.

\section{GenAI Usage Disclosure}
We used ChatGPT strictly for language polishing and text refinement of author-written content to improve the manuscript's readability. No generative AI tools were used for research design, code or data generation, experiments, analysis, results, conclusions, or figure creation. All authors reviewed the resulting text and assume full responsibility for the completeness, accuracy, and integrity of this manuscript.

% \begin{acks}
% \end{acks}

\bibliographystyle{ACM-Reference-Format}
\balance
\bibliography{sample-base}


\begin{thebibliography}{36}

%%% except for \shownote{} and \showURL{}.  The latter does
%%% not use final punctuation, in order to avoid confusing it with
%%% the Web address.

\ifx \showCODEN    \undefined \def \showCODEN     #1{\unskip}     \fi
\ifx \showISBNx    \undefined \def \showISBNx     #1{\unskip}     \fi
\ifx \showISBNxiii \undefined \def \showISBNxiii  #1{\unskip}     \fi
\ifx \showISSN     \undefined \def \showISSN      #1{\unskip}     \fi
\ifx \showLCCN     \undefined \def \showLCCN      #1{\unskip}     \fi
\ifx \shownote     \undefined \def \shownote      #1{#1}          \fi
\ifx \showarticletitle \undefined \def \showarticletitle #1{#1}   \fi
\ifx \showURL      \undefined \def \showURL       {\relax}        \fi
% invisible to TeX
\providecommand\bibfield[2]{#2}
\providecommand\bibinfo[2]{#2}
\providecommand\natexlab[1]{#1}
\providecommand\showeprint[2][]{arXiv:#2}
\makeatletter
\@ifundefined{NAT@parse@date}{}{\let\NAT@parse@date@orig\NAT@parse@date}
\@ifundefined{NAT@parse@date}{}{\def\NAT@parse@date#1#2#3#4#5#6@@{\NAT@parse@date@orig#1#2#3#4#5#6@@\def\NAT@tempyear{0000}\def\NAT@tempexlab{{?}}\ifx\NAT@year\NAT@tempyear\ifx\NAT@exlab\NAT@tempexlab\def\NAT@date{[n.\,d.]}\else\edef\NAT@date{[n.\,d.]\NAT@exlab}\fi\fi}}
\makeatother

\bibitem[Ai et~al\mbox{.}(2024)]%
        {ai2024improve}
\bibfield{author}{\bibinfo{person}{Meng Ai}, \bibinfo{person}{Zhuo Chen}, \bibinfo{person}{Jibin Wang}, \bibinfo{person}{Jing Shang}, \bibinfo{person}{Tao Tao}, {and} \bibinfo{person}{Zhen Li}.} \bibinfo{year}{2024}\natexlab{}.
\newblock \showarticletitle{Improve roi with causal learning and conformal prediction}. In \bibinfo{booktitle}{\emph{2024 IEEE 40th International Conference on Data Engineering (ICDE)}}. IEEE, \bibinfo{pages}{598--610}.
\newblock


\bibitem[Bi et~al\mbox{.}(2022)]%
        {bi2022mtrec}
\bibfield{author}{\bibinfo{person}{Qiwei Bi}, \bibinfo{person}{Jian Li}, \bibinfo{person}{Lifeng Shang}, \bibinfo{person}{Xin Jiang}, \bibinfo{person}{Qun Liu}, {and} \bibinfo{person}{Hanfang Yang}.} \bibinfo{year}{2022}\natexlab{}.
\newblock \showarticletitle{Mtrec: Multi-task learning over bert for news recommendation}. In \bibinfo{booktitle}{\emph{Findings of the association for computational linguistics: ACL 2022}}. \bibinfo{pages}{2663--2669}.
\newblock


\bibitem[Caruana(1997)]%
        {caruana1997multitask}
\bibfield{author}{\bibinfo{person}{Rich Caruana}.} \bibinfo{year}{1997}\natexlab{}.
\newblock \showarticletitle{Multitask learning}.
\newblock \bibinfo{journal}{\emph{Machine learning}}  \bibinfo{volume}{28} (\bibinfo{year}{1997}), \bibinfo{pages}{41--75}.
\newblock


\bibitem[Chang et~al\mbox{.}(2023a)]%
        {chang2023twin}
\bibfield{author}{\bibinfo{person}{Jianxin Chang}, \bibinfo{person}{Chenbin Zhang}, \bibinfo{person}{Zhiyi Fu}, \bibinfo{person}{Xiaoxue Zang}, \bibinfo{person}{Lin Guan}, \bibinfo{person}{Jing Lu}, \bibinfo{person}{Yiqun Hui}, \bibinfo{person}{Dewei Leng}, \bibinfo{person}{Yanan Niu}, \bibinfo{person}{Yang Song}, {et~al\mbox{.}}} \bibinfo{year}{2023}\natexlab{a}.
\newblock \showarticletitle{TWIN: TWo-stage interest network for lifelong user behavior modeling in CTR prediction at kuaishou}. In \bibinfo{booktitle}{\emph{Proceedings of the 29th ACM SIGKDD Conference on Knowledge Discovery and Data Mining}}. \bibinfo{pages}{3785--3794}.
\newblock


\bibitem[Chang et~al\mbox{.}(2023b)]%
        {ppnet2023}
\bibfield{author}{\bibinfo{person}{Jianxin Chang}, \bibinfo{person}{Chenbin Zhang}, \bibinfo{person}{Yiqun Hui}, \bibinfo{person}{Dewei Leng}, \bibinfo{person}{Yanan Niu}, \bibinfo{person}{Yang Song}, {and} \bibinfo{person}{Kun Gai}.} \bibinfo{year}{2023}\natexlab{b}.
\newblock \showarticletitle{PEPNet: Parameter and Embedding Personalized Network for Infusing with Personalized Prior Information}. In \bibinfo{booktitle}{\emph{Proceedings of the 29th ACM SIGKDD Conference on Knowledge Discovery and Data Mining}} (Long Beach, CA, USA) \emph{(\bibinfo{series}{KDD '23})}. \bibinfo{publisher}{Association for Computing Machinery}, \bibinfo{address}{New York, NY, USA}, \bibinfo{pages}{3795--3804}.
\newblock


\bibitem[Curth and Van~der Schaar(2021)]%
        {curth2021inductive}
\bibfield{author}{\bibinfo{person}{Alicia Curth} {and} \bibinfo{person}{Mihaela Van~der Schaar}.} \bibinfo{year}{2021}\natexlab{}.
\newblock \showarticletitle{On inductive biases for heterogeneous treatment effect estimation}.
\newblock \bibinfo{journal}{\emph{Advances in Neural Information Processing Systems}}  \bibinfo{volume}{34} (\bibinfo{year}{2021}), \bibinfo{pages}{15883--15894}.
\newblock


\bibitem[Fu et~al\mbox{.}(2025)]%
        {fu2025unified}
\bibfield{author}{\bibinfo{person}{Zichuan Fu}, \bibinfo{person}{Xiangyang Li}, \bibinfo{person}{Chuhan Wu}, \bibinfo{person}{Yichao Wang}, \bibinfo{person}{Kuicai Dong}, \bibinfo{person}{Xiangyu Zhao}, \bibinfo{person}{Mengchen Zhao}, \bibinfo{person}{Huifeng Guo}, {and} \bibinfo{person}{Ruiming Tang}.} \bibinfo{year}{2025}\natexlab{}.
\newblock \showarticletitle{A unified framework for multi-domain ctr prediction via large language models}.
\newblock \bibinfo{journal}{\emph{ACM Transactions on Information Systems}} \bibinfo{volume}{43}, \bibinfo{number}{5} (\bibinfo{year}{2025}), \bibinfo{pages}{1--33}.
\newblock


\bibitem[Guo et~al\mbox{.}(2017)]%
        {guo2017deepfm}
\bibfield{author}{\bibinfo{person}{Huifeng Guo}, \bibinfo{person}{TANG Ruiming}, \bibinfo{person}{Yunming Ye}, \bibinfo{person}{Zhenguo Li}, {and} \bibinfo{person}{Xiuqiang He}.} \bibinfo{year}{2017}\natexlab{}.
\newblock \showarticletitle{DeepFM: A Factorization-Machine based Neural Network for CTR Prediction}. In \bibinfo{booktitle}{\emph{Proceedings of the Twenty-Sixth International Joint Conference on Artificial Intelligence}}. International Joint Conferences on Artificial Intelligence Organization.
\newblock


\bibitem[He et~al\mbox{.}(2024a)]%
        {he2024rankability}
\bibfield{author}{\bibinfo{person}{Bowei He}, \bibinfo{person}{Yunpeng Weng}, \bibinfo{person}{Xing Tang}, \bibinfo{person}{Ziqiang Cui}, \bibinfo{person}{Zexu Sun}, \bibinfo{person}{Liang Chen}, \bibinfo{person}{Xiuqiang He}, {and} \bibinfo{person}{Chen Ma}.} \bibinfo{year}{2024}\natexlab{a}.
\newblock \showarticletitle{Rankability-enhanced revenue uplift modeling framework for online marketing}. In \bibinfo{booktitle}{\emph{Proceedings of the 30th ACM SIGKDD Conference on Knowledge Discovery and Data Mining}}. \bibinfo{pages}{5093--5104}.
\newblock


\bibitem[He et~al\mbox{.}(2024b)]%
        {he2024rankabilityenhancedrevenueupliftmodeling}
\bibfield{author}{\bibinfo{person}{Bowei He}, \bibinfo{person}{Yunpeng Weng}, \bibinfo{person}{Xing Tang}, \bibinfo{person}{Ziqiang Cui}, \bibinfo{person}{Zexu Sun}, \bibinfo{person}{Liang Chen}, \bibinfo{person}{Xiuqiang He}, {and} \bibinfo{person}{Chen Ma}.} \bibinfo{year}{2024}\natexlab{b}.
\newblock \bibinfo{title}{Rankability-enhanced Revenue Uplift Modeling Framework for Online Marketing}.
\newblock
\showeprint[arxiv]{2405.15301}~[cs.LG]


\bibitem[Jiang et~al\mbox{.}(2024)]%
        {jiang2024automatic}
\bibfield{author}{\bibinfo{person}{Shen Jiang}, \bibinfo{person}{Guanghui Zhu}, \bibinfo{person}{Yue Wang}, \bibinfo{person}{Chunfeng Yuan}, {and} \bibinfo{person}{Yihua Huang}.} \bibinfo{year}{2024}\natexlab{}.
\newblock \showarticletitle{Automatic Multi-Task Learning Framework with Neural Architecture Search in Recommendations}. In \bibinfo{booktitle}{\emph{Proceedings of the 30th ACM SIGKDD Conference on Knowledge Discovery and Data Mining}}. \bibinfo{pages}{1290--1300}.
\newblock


\bibitem[Ke et~al\mbox{.}(2021)]%
        {ke2021addressing}
\bibfield{author}{\bibinfo{person}{Wenwei Ke}, \bibinfo{person}{Chuanren Liu}, \bibinfo{person}{Xiangfu Shi}, \bibinfo{person}{Yiqiao Dai}, \bibinfo{person}{Philip~S Yu}, {and} \bibinfo{person}{Xiaoqiang Zhu}.} \bibinfo{year}{2021}\natexlab{}.
\newblock \showarticletitle{Addressing exposure bias in uplift modeling for large-scale online advertising}. In \bibinfo{booktitle}{\emph{2021 IEEE International Conference on Data Mining (ICDM)}}. IEEE, \bibinfo{pages}{1156--1161}.
\newblock


\bibitem[K{\"u}nzel et~al\mbox{.}(2019)]%
        {kunzel2019metalearners}
\bibfield{author}{\bibinfo{person}{S{\"o}ren~R K{\"u}nzel}, \bibinfo{person}{Jasjeet~S Sekhon}, \bibinfo{person}{Peter~J Bickel}, {and} \bibinfo{person}{Bin Yu}.} \bibinfo{year}{2019}\natexlab{}.
\newblock \showarticletitle{Metalearners for estimating heterogeneous treatment effects using machine learning}.
\newblock \bibinfo{journal}{\emph{Proceedings of the national academy of sciences}} \bibinfo{volume}{116}, \bibinfo{number}{10} (\bibinfo{year}{2019}), \bibinfo{pages}{4156--4165}.
\newblock


\bibitem[Li et~al\mbox{.}(2023b)]%
        {li2023adatt}
\bibfield{author}{\bibinfo{person}{Danwei Li}, \bibinfo{person}{Zhengyu Zhang}, \bibinfo{person}{Siyang Yuan}, \bibinfo{person}{Mingze Gao}, \bibinfo{person}{Weilin Zhang}, \bibinfo{person}{Chaofei Yang}, \bibinfo{person}{Xi Liu}, {and} \bibinfo{person}{Jiyan Yang}.} \bibinfo{year}{2023}\natexlab{b}.
\newblock \showarticletitle{Adatt: Adaptive task-to-task fusion network for multitask learning in recommendations}. In \bibinfo{booktitle}{\emph{Proceedings of the 29th ACM SIGKDD Conference on Knowledge Discovery and Data Mining}}. \bibinfo{pages}{4370--4379}.
\newblock


\bibitem[Li et~al\mbox{.}(2023a)]%
        {li2023removing}
\bibfield{author}{\bibinfo{person}{Haoxuan Li}, \bibinfo{person}{Kunhan Wu}, \bibinfo{person}{Chunyuan Zheng}, \bibinfo{person}{Yanghao Xiao}, \bibinfo{person}{Hao Wang}, \bibinfo{person}{Zhi Geng}, \bibinfo{person}{Fuli Feng}, \bibinfo{person}{Xiangnan He}, {and} \bibinfo{person}{Peng Wu}.} \bibinfo{year}{2023}\natexlab{a}.
\newblock \showarticletitle{Removing hidden confounding in recommendation: a unified multi-task learning approach}.
\newblock \bibinfo{journal}{\emph{Advances in Neural Information Processing Systems}}  \bibinfo{volume}{36} (\bibinfo{year}{2023}), \bibinfo{pages}{54614--54626}.
\newblock


\bibitem[Liu et~al\mbox{.}(2023)]%
        {liu2023explicitfeatureinteractionawareuplift}
\bibfield{author}{\bibinfo{person}{Dugang Liu}, \bibinfo{person}{Xing Tang}, \bibinfo{person}{Han Gao}, \bibinfo{person}{Fuyuan Lyu}, {and} \bibinfo{person}{Xiuqiang He}.} \bibinfo{year}{2023}\natexlab{}.
\newblock \bibinfo{title}{Explicit Feature Interaction-aware Uplift Network for Online Marketing}.
\newblock
\showeprint[arxiv]{2306.00315}~[cs.LG]


\bibitem[Ma et~al\mbox{.}(2018b)]%
        {ma2018modeling}
\bibfield{author}{\bibinfo{person}{Jiaqi Ma}, \bibinfo{person}{Zhe Zhao}, \bibinfo{person}{Xinyang Yi}, \bibinfo{person}{Jilin Chen}, \bibinfo{person}{Lichan Hong}, {and} \bibinfo{person}{Ed~H Chi}.} \bibinfo{year}{2018}\natexlab{b}.
\newblock \showarticletitle{Modeling task relationships in multi-task learning with multi-gate mixture-of-experts}. In \bibinfo{booktitle}{\emph{Proceedings of the 24th ACM SIGKDD international conference on knowledge discovery \& data mining}}. \bibinfo{pages}{1930--1939}.
\newblock


\bibitem[Ma et~al\mbox{.}(2018c)]%
        {MMOE2018}
\bibfield{author}{\bibinfo{person}{Jiaqi Ma}, \bibinfo{person}{Zhe Zhao}, \bibinfo{person}{Xinyang Yi}, \bibinfo{person}{Jilin Chen}, \bibinfo{person}{Lichan Hong}, {and} \bibinfo{person}{Ed~H. Chi}.} \bibinfo{year}{2018}\natexlab{c}.
\newblock \showarticletitle{Modeling Task Relationships in Multi-task Learning with Multi-gate Mixture-of-Experts}.
\newblock


\bibitem[Ma et~al\mbox{.}(2018a)]%
        {ma2018entire}
\bibfield{author}{\bibinfo{person}{Xiao Ma}, \bibinfo{person}{Liqin Zhao}, \bibinfo{person}{Guan Huang}, \bibinfo{person}{Zhi Wang}, \bibinfo{person}{Zelin Hu}, \bibinfo{person}{Xiaoqiang Zhu}, {and} \bibinfo{person}{Kun Gai}.} \bibinfo{year}{2018}\natexlab{a}.
\newblock \showarticletitle{Entire space multi-task model: An effective approach for estimating post-click conversion rate}. In \bibinfo{booktitle}{\emph{The 41st International ACM SIGIR Conference on Research \& Development in Information Retrieval}}. \bibinfo{pages}{1137--1140}.
\newblock


\bibitem[Mao et~al\mbox{.}(2023)]%
        {mao2023finalmlp}
\bibfield{author}{\bibinfo{person}{Kelong Mao}, \bibinfo{person}{Jieming Zhu}, \bibinfo{person}{Liangcai Su}, \bibinfo{person}{Guohao Cai}, \bibinfo{person}{Yuru Li}, {and} \bibinfo{person}{Zhenhua Dong}.} \bibinfo{year}{2023}\natexlab{}.
\newblock \showarticletitle{FinalMLP: an enhanced two-stream MLP model for CTR prediction}. In \bibinfo{booktitle}{\emph{Proceedings of the AAAI conference on artificial intelligence}}, Vol.~\bibinfo{volume}{37}. \bibinfo{pages}{4552--4560}.
\newblock


\bibitem[Meng et~al\mbox{.}(2024)]%
        {meng2024coarsetofinedynamicupliftmodeling}
\bibfield{author}{\bibinfo{person}{Chang Meng}, \bibinfo{person}{Chenhao Zhai}, \bibinfo{person}{Xueliang Wang}, \bibinfo{person}{Shuchang Liu}, \bibinfo{person}{Xiaoqiang Feng}, \bibinfo{person}{Lantao Hu}, \bibinfo{person}{Xiu Li}, \bibinfo{person}{Han Li}, {and} \bibinfo{person}{Kun Gai}.} \bibinfo{year}{2024}\natexlab{}.
\newblock \bibinfo{title}{Coarse-to-fine Dynamic Uplift Modeling for Real-time Video Recommendation}.
\newblock
\showeprint[arxiv]{2410.16755}~[cs.IR]


\bibitem[Nie et~al\mbox{.}(2021)]%
        {nie2021vcnet}
\bibfield{author}{\bibinfo{person}{Lizhen Nie}, \bibinfo{person}{Mao Ye}, \bibinfo{person}{Qiang Liu}, {and} \bibinfo{person}{Dan Nicolae}.} \bibinfo{year}{2021}\natexlab{}.
\newblock \showarticletitle{Vcnet and functional targeted regularization for learning causal effects of continuous treatments}.
\newblock \bibinfo{journal}{\emph{arXiv preprint arXiv:2103.07861}} (\bibinfo{year}{2021}).
\newblock


\bibitem[Rubin(2005)]%
        {Rubin01032005}
\bibfield{author}{\bibinfo{person}{Donald~B Rubin}.} \bibinfo{year}{2005}\natexlab{}.
\newblock \showarticletitle{Causal Inference Using Potential Outcomes}.
\newblock \bibinfo{journal}{\emph{J. Amer. Statist. Assoc.}} \bibinfo{volume}{100}, \bibinfo{number}{469} (\bibinfo{year}{2005}), \bibinfo{pages}{322--331}.
\newblock


\bibitem[Schwab et~al\mbox{.}(2020)]%
        {schwab2020learning}
\bibfield{author}{\bibinfo{person}{Patrick Schwab}, \bibinfo{person}{Lorenz Linhardt}, \bibinfo{person}{Stefan Bauer}, \bibinfo{person}{Joachim~M Buhmann}, {and} \bibinfo{person}{Walter Karlen}.} \bibinfo{year}{2020}\natexlab{}.
\newblock \showarticletitle{Learning counterfactual representations for estimating individual dose-response curves}. In \bibinfo{booktitle}{\emph{Proceedings of the AAAI Conference on Artificial Intelligence}}, Vol.~\bibinfo{volume}{34}. \bibinfo{pages}{5612--5619}.
\newblock


\bibitem[Shalit et~al\mbox{.}(2017)]%
        {shalit2017estimating}
\bibfield{author}{\bibinfo{person}{Uri Shalit}, \bibinfo{person}{Fredrik~D Johansson}, {and} \bibinfo{person}{David Sontag}.} \bibinfo{year}{2017}\natexlab{}.
\newblock \showarticletitle{Estimating individual treatment effect: generalization bounds and algorithms}. In \bibinfo{booktitle}{\emph{International conference on machine learning}}. PMLR, \bibinfo{pages}{3076--3085}.
\newblock


\bibitem[Su et~al\mbox{.}(2024)]%
        {su2024stem}
\bibfield{author}{\bibinfo{person}{Liangcai Su}, \bibinfo{person}{Junwei Pan}, \bibinfo{person}{Ximei Wang}, \bibinfo{person}{Xi Xiao}, \bibinfo{person}{Shijie Quan}, \bibinfo{person}{Xihua Chen}, {and} \bibinfo{person}{Jie Jiang}.} \bibinfo{year}{2024}\natexlab{}.
\newblock \showarticletitle{STEM: unleashing the power of embeddings for multi-task recommendation}. In \bibinfo{booktitle}{\emph{Proceedings of the AAAI Conference on Artificial Intelligence}}, Vol.~\bibinfo{volume}{38}. \bibinfo{pages}{9002--9010}.
\newblock


\bibitem[Tang et~al\mbox{.}(2020)]%
        {tang2020progressive}
\bibfield{author}{\bibinfo{person}{Hongyan Tang}, \bibinfo{person}{Junning Liu}, \bibinfo{person}{Ming Zhao}, {and} \bibinfo{person}{Xudong Gong}.} \bibinfo{year}{2020}\natexlab{}.
\newblock \showarticletitle{Progressive layered extraction (ple): A novel multi-task learning (mtl) model for personalized recommendations}. In \bibinfo{booktitle}{\emph{Proceedings of the 14th ACM conference on recommender systems}}. \bibinfo{pages}{269--278}.
\newblock


\bibitem[Tao et~al\mbox{.}(2023)]%
        {tao2023event}
\bibfield{author}{\bibinfo{person}{Wanjie Tao}, \bibinfo{person}{Huihui Liu}, \bibinfo{person}{Xuqi Li}, \bibinfo{person}{Qun Dai}, \bibinfo{person}{Hong Wen}, {and} \bibinfo{person}{Zulong Chen}.} \bibinfo{year}{2023}\natexlab{}.
\newblock \showarticletitle{Event-Aware Adaptive Clustering Uplift Network for Insurance Creative Ranking}. In \bibinfo{booktitle}{\emph{Proceedings of the 46th International ACM SIGIR Conference on Research and Development in Information Retrieval}}. \bibinfo{pages}{1966--1970}.
\newblock


\bibitem[Wan et~al\mbox{.}(2022)]%
        {wan2022gcfgeneralizedcausalforest}
\bibfield{author}{\bibinfo{person}{Shu Wan}, \bibinfo{person}{Chen Zheng}, \bibinfo{person}{Zhonggen Sun}, \bibinfo{person}{Mengfan Xu}, \bibinfo{person}{Xiaoqing Yang}, \bibinfo{person}{Hongtu Zhu}, {and} \bibinfo{person}{Jiecheng Guo}.} \bibinfo{year}{2022}\natexlab{}.
\newblock \bibinfo{title}{GCF: Generalized Causal Forest for Heterogeneous Treatment Effect Estimation in Online Marketplace}.
\newblock
\showeprint[arxiv]{2203.10975}~[stat.ML]
\urldef\tempurl%
\url{https://arxiv.org/abs/2203.10975}
\showURL{%
\tempurl}


\bibitem[Wang et~al\mbox{.}(2025)]%
        {wang2025inter}
\bibfield{author}{\bibinfo{person}{Fan Wang}, \bibinfo{person}{Lianyong Qi}, \bibinfo{person}{Weiming Liu}, \bibinfo{person}{Bowen Yu}, \bibinfo{person}{Jintao Chen}, {and} \bibinfo{person}{Yanwei Xu}.} \bibinfo{year}{2025}\natexlab{}.
\newblock \showarticletitle{Inter-and intra-similarity preserved counterfactual incentive effect estimation for recommendation systems}.
\newblock \bibinfo{journal}{\emph{ACM Transactions on Information Systems}} \bibinfo{volume}{43}, \bibinfo{number}{6} (\bibinfo{year}{2025}), \bibinfo{pages}{1--24}.
\newblock


\bibitem[Wang et~al\mbox{.}(2022)]%
        {wang2022enhancing}
\bibfield{author}{\bibinfo{person}{Fangye Wang}, \bibinfo{person}{Yingxu Wang}, \bibinfo{person}{Dongsheng Li}, \bibinfo{person}{Hansu Gu}, \bibinfo{person}{Tun Lu}, \bibinfo{person}{Peng Zhang}, {and} \bibinfo{person}{Ning Gu}.} \bibinfo{year}{2022}\natexlab{}.
\newblock \showarticletitle{Enhancing CTR prediction with context-aware feature representation learning}. In \bibinfo{booktitle}{\emph{Proceedings of the 45th International ACM SIGIR Conference on Research and Development in Information Retrieval}}. \bibinfo{pages}{343--352}.
\newblock


\bibitem[Wang et~al\mbox{.}(2023)]%
        {wang2023single}
\bibfield{author}{\bibinfo{person}{Yejing Wang}, \bibinfo{person}{Zhaocheng Du}, \bibinfo{person}{Xiangyu Zhao}, \bibinfo{person}{Bo Chen}, \bibinfo{person}{Huifeng Guo}, \bibinfo{person}{Ruiming Tang}, {and} \bibinfo{person}{Zhenhua Dong}.} \bibinfo{year}{2023}\natexlab{}.
\newblock \showarticletitle{Single-shot feature selection for multi-task recommendations}. In \bibinfo{booktitle}{\emph{Proceedings of the 46th International ACM SIGIR Conference on Research and Development in Information Retrieval}}. \bibinfo{pages}{341--351}.
\newblock


\bibitem[Yang et~al\mbox{.}(2021)]%
        {yang2021multi}
\bibfield{author}{\bibinfo{person}{Haizhi Yang}, \bibinfo{person}{Tengyun Wang}, \bibinfo{person}{Xiaoli Tang}, \bibinfo{person}{Qianyu Li}, \bibinfo{person}{Yueyue Shi}, \bibinfo{person}{Siyu Jiang}, \bibinfo{person}{Han Yu}, {and} \bibinfo{person}{Hengjie Song}.} \bibinfo{year}{2021}\natexlab{}.
\newblock \showarticletitle{Multi-task learning for bias-free joint ctr prediction and market price modeling in online advertising}. In \bibinfo{booktitle}{\emph{Proceedings of the 30th ACM International Conference on Information \& Knowledge Management}}. \bibinfo{pages}{2291--2300}.
\newblock


\bibitem[Zhang et~al\mbox{.}(2024)]%
        {zhang2024temporal}
\bibfield{author}{\bibinfo{person}{Xin Zhang}, \bibinfo{person}{Kai Wang}, \bibinfo{person}{Zengmao Wang}, \bibinfo{person}{Bo Du}, \bibinfo{person}{Shiwei Zhao}, \bibinfo{person}{Runze Wu}, \bibinfo{person}{Xudong Shen}, \bibinfo{person}{Tangjie Lv}, {and} \bibinfo{person}{Changjie Fan}.} \bibinfo{year}{2024}\natexlab{}.
\newblock \showarticletitle{Temporal Uplift Modeling for Online Marketing}. In \bibinfo{booktitle}{\emph{Proceedings of the 30th ACM SIGKDD Conference on Knowledge Discovery and Data Mining}}. \bibinfo{pages}{6247--6256}.
\newblock


\bibitem[Zhao et~al\mbox{.}(2019)]%
        {zhao2019recommending}
\bibfield{author}{\bibinfo{person}{Zhe Zhao}, \bibinfo{person}{Lichan Hong}, \bibinfo{person}{Li Wei}, \bibinfo{person}{Jilin Chen}, \bibinfo{person}{Aniruddh Nath}, \bibinfo{person}{Shawn Andrews}, \bibinfo{person}{Aditee Kumthekar}, \bibinfo{person}{Maheswaran Sathiamoorthy}, \bibinfo{person}{Xinyang Yi}, {and} \bibinfo{person}{Ed Chi}.} \bibinfo{year}{2019}\natexlab{}.
\newblock \showarticletitle{Recommending what video to watch next: a multitask ranking system}. In \bibinfo{booktitle}{\emph{Proceedings of the 13th ACM conference on recommender systems}}. \bibinfo{pages}{43--51}.
\newblock


\bibitem[Zhong et~al\mbox{.}(2022)]%
        {zhong2022descn}
\bibfield{author}{\bibinfo{person}{Kailiang Zhong}, \bibinfo{person}{Fengtong Xiao}, \bibinfo{person}{Yan Ren}, \bibinfo{person}{Yaorong Liang}, \bibinfo{person}{Wenqing Yao}, \bibinfo{person}{Xiaofeng Yang}, {and} \bibinfo{person}{Ling Cen}.} \bibinfo{year}{2022}\natexlab{}.
\newblock \showarticletitle{Descn: Deep entire space cross networks for individual treatment effect estimation}. In \bibinfo{booktitle}{\emph{Proceedings of the 28th ACM SIGKDD conference on knowledge discovery and data mining}}. \bibinfo{pages}{4612--4620}.
\newblock


\end{thebibliography}

% \input{samples/chapters/appendix}









\end{document}